%% file: counting.tex
\documentclass[12pt]{article}
\usepackage{graphicx}
\usepackage{amsmath}
\usepackage{amssymb}
\usepackage{amsthm}
\usepackage{stmaryrd}
\usepackage{mathpartir}
\usepackage{cmll}
\usepackage{todonotes}
\usepackage{bm}
\usepackage{mathrsfs}
\usepackage{a4wide}
\usepackage[all]{xy}
\usepackage{color}
\definecolor{grey}{rgb}{.7,.7,.7}
\newcommand{\grey}[1]{{\color{grey}#1}}
\newcommand{\hide}[1]{}
\newtheorem{proposition}{Proposition}
\newtheorem{definition}[proposition]{Definition}
\newtheorem{theorem}[proposition]{Theorem}
\newtheorem{lemma}[proposition]{Lemma}
\newtheorem{corollary}[proposition]{Corollary}
\newtheorem*{theorem*}{Theorem}

\theoremstyle{remark}
\newtheorem{example}[proposition]{Example}

\input{macros.tex}

\usepackage{graphicx}
 \usepackage[all]{xy}

\def\profto{\!\!\!\xymatrix@C-.75pc{\ar[r]|-{\! +\!} &}\!\!\! }

  \newcommand{\spanpl}[5]{{\xymatrix{
    & {#3}\ar@{_{(}->}[dl]_{#2}\ar[dr]^-{#4} &\\
      {#1} && {#5} 
  }}}

\def\mxxth{\mathsurround=0pt}
\dimendef\dimenxx=0
\def\openup{\afterassignment\xxpenup\dimenxx=}
\def\xxpenup{\advance\ltoeskip\dimenxx
  \advance\baselineskip\dimenxx \advance\ltoeskiplimit\dimenxx}

\newif\ifdtxxp
\def\displxxy{\global\dtxxptrue \openup1\jot \mxxth
  \everycr{\noalign{\ifdtxxp \global\dtxxpfalse
      \vskip-\ltoeskiplimit \vskip\normallineskiplimit
      \else \penalty\interdisplaylinepenalty \fi}}}
\def\displaylines#1{\displxxy
  \halign{\hbox to\displaywidth{$\hfil\displaystyle##\hfil$}\crcr
      #1\crcr}}

\newskip\mycntring \mycntring=0pt plus 1000pt minus 1000pt

\def\leqalignno#1{\displxxy \tabskip=\mycntring
  \halign to\displaywidth{\hfil$\displaystyle{##}$\tabskip=0pt
      &$\displaystyle{{}##}$\hfil\tabskip=\mycntring
      &\kern-\displaywidth\rlap{$##$}\tabskip=\displaywidth\crcr
      #1\crcr}}

\newdir{C}{%
!/4.5pt/@{(}*:(1,-.2)@{}*:(1,+.2)@_{}}

\newdir{|>}{%
!/4.5pt/@{|}*:(1,-.2)@{>}*:(1,+.2)@_{>}}

\newcommand{\card}[1]{|#1|}

\newdir{pb}{:(1,-1)@^{|-}}
\def\pb#1{\save[]+<16 pt,0 pt>:a(#1)\ar@{pb{}}[]\restore}

\begin{document}
\title{Learning to Count up to Symmetry}
\author{Pierre Clairambault\\
Univ Lyon, EnsL, UCBL, CNRS,  LIP, F-69342, LYON Cedex 07, France}
\date{}
\maketitle
\begin{abstract}
In this paper we develop the theory of how to \emph{count}, in thin concurrent games,
the configurations of a strategy witnessing that it reaches a certain configuration
of the game. This plays a central role in many recent developments in concurrent
games, whenever one aims to relate concurrent strategies with weighted relational
models.

The difficulty, of course, is symmetry: in the presence of symmetry many
configurations of the strategy are, morally, different instances of the \emph{same},
only differing on the inessential choice of copy indices. How do we know which ones to
count? The purpose of the paper is to clarify that, uncovering many strange phenomena
and fascinating pathological examples along the way.

To illustrate the results, we show that a collapse operation to a simple weighted
relational model simply counting witnesses is preserved under composition,
provided the strategies involved do not deadlock.
\end{abstract}

\section{Introduction}

\emph{Thin concurrent games} \cite{cg2} are a complex but powerful setting for truly
concurrent game semantics; one of the latest iterations of a long line of work
\cite{AbramskyM99,MelliesM07,lics11} on game semantics
questioning the premise that a play should be a total chronological ordering.
They are very expressive, able to express various languages
both pure \cite{lics15} and stateful \cite{cg2}; including with various quantitative
aspects \cite{lics18,popl20}. One strength of concurrent games in general is the
clean link they offer with relational-like semantics: a strategy may (slightly
naively) be seen as a collection of points of the \emph{web} (in the sense of relational
semantics) enriched with causal information. This enables a clean connection with the
relational model, which served as basis \emph{e.g.} for Melli\`es' fully complete
model of linear logic \cite{ag4} (see also \cite{mall}).

Now, relational semantics as well can be enriched with quantitative information; this
is the basis for \emph{probabilistic coherence spaces} \cite{DanosE11}. Probabilistic
coherence spaces are obtained via a biorthogonality construction on top of the
relational model \emph{weighted} by elements of $\overline{\mathbb{R}_+}$, the
completion of non-negative reals $\mathbb{R}_+$ with a point at infinity. Instead of
merely relations, morphisms from set $A$ to set $B$ are then matrices
\[
(\alpha_{a, b})_{(a, b) \in A\times B} \in \overline{\mathbb{R}_+}^{A\times B}
\]
composed via the potentially infinite matrix multiplication formula
\begin{eqnarray}
(\beta \circ \alpha)_{a,c} &=& \sum_{b \in B} \alpha_{a, b} \cdot
\beta_{b,c}\,.\label{eq0}
\end{eqnarray}

Beyond real scalars, more generally one can construct a \emph{weighted relational
model} parametrized by certain semirings
\cite{DBLP:conf/lics/LairdMMP13}. Adding typing information one goes beyond
semirings, for instance the adequate model for the quantum $\lambda$-calculus of
\cite{PaganiSV14} uses weights from the category of finite dimensional Hilbert spaces
and completely positive maps. 

Above, we mentioned a \emph{collapse} from concurrent games to the relational model.
Does it hold with quantitative information? Such results appear in the
literature \cite{lics18, popl20} -- though we shall see in this paper that the
definition of this collapse in \cite{lics18} is not quite right. This seemingly
simple question holds some surprises. This is the question that this paper solves;
detailing the basis for part of \cite{popl20}, and identifying and correcting the
mistake in \cite{lics18}.

As a matter of fact, the difficulty is not in handling the \emph{weights}, but in
\emph{listing the right witnesses}: if \eqref{eq0} originates in a bijection between
witnesses, then provided this bijection preserves the weights (and it will be
generated in such a way that it does), it follows that adding weights is relatively
painless. On the other hand, coming up with the right notion of witnesses is really
hard. Indeed, in the presence of replication of resources, configurations in
strategies are countably duplicated, so it is meaningless to sum over all of those as
one does without symmetry. What are, then, the right witnesses? Symmetry classes of
configurations? Something else? In this paper we give the answer, and illustrate it
with a proof of a formula like \eqref{eq0} for a simple weighted relational model
simply counting witnesses.  

We shall see that the appealingly simple idea of \cite{lics18} to use
symmetry classes of configurations as witnesses is, in general, wrong. We give a more
refined notion of witnesses, taking advantage of the split of the symmetry into
\emph{positive} and \emph{negative} reindexings offered by thin concurrent games
\cite{cg2}. This lets us solve the problem, but with the cost of adding a new
condition to thin concurrent games called \emph{representability}, which states the
existence, for every symmetry class, of a \emph{canonical} representative on which
the symmetry decomposes neatly into a positive and negative parts.

\paragraph{Outline.}
The structure of the paper is as follows. In Section \ref{sec:prelim} we fix the
notations for thin concurrent games used in this paper and recall a few notions. In
Section \ref{sec:quant_coll} we give a technical explanation of the problem and its
difficulties. In Section \ref{sec:canon_repr} we introduce the new notions of
\emph{canonicity} and \emph{representability}. In Section \ref{sec:quant_coll_final}
we give the central contribution of the paper, the proof of \eqref{eq0}.
Finally, in Section \ref{sec:epilogue} we give a few ending remarks. 

\section{Preliminaries}
\label{sec:prelim}

\subsection{Notations and terminology}

In this paper, we assume some familiarity with concurrent games, and more precisely
with \emph{thin concurrent games} \cite{cg2}. Let us fix a few
conventions for notations and terminology.

By \emph{strategy} we will always mean $\sim$-strategy in the sense of \cite{cg2}. We
will sometimes refer to \emph{pre-$\sim$-strategies}, which must be understood as in
\cite{cg2}. If $\sigma : S \to A^\perp \parallel B$ is a strategy from $A$ to $B$, we
write $\sigma : A \stackrel{S}{\to} B$. We often use $x^S, y^S, \dots$ to range over
configurations of $S$, with $S$ as a superscript. If $x^S \in \conf{S}$, we take the
convention that 
\[
\sigma x^S = x^S_A \parallel x^S_B\,,
\]
in the paper we will use $x^S_A \in \conf{A}$ and $x^S_B \in \conf{B}$ without further
introduction.

If $A$ is a tcg, we write $\sym_A$ for its symmetry, and $\theta : x \sym_A y$ if
the bijection $\theta : x \simeq y$ is in $\sym_A$ -- in which case we say that
$\theta$ is \emph{a symmetry}. For $x, y \in \conf{A}$, we write $x \sym_A y$ for the
induced equivalence relation. We use similar notations for the positive and negative
sub-symmetries, with $\sym_A^+$ for the positive and $\sym_A^-$ for the negative.
We use for symmetries on strategies similar notations as for configurations. For
$\sigma : A \stackrel{S}{\to} B$, we often tag symmetries in $S$ with $S$, as in
$\varphi^S : x^S \sym_S y^S$. Then, we write $\varphi^S_A : x^S_A \sym_A y^S_A$ and
$\varphi^S_B : x^S_B \sym_B y^S_B$.

In diagrams, dotted lines signify immediate causal links
in the game, whereas $\imc$ means immediate causality in the strategy. If the
direction of causal links is unspecified (\emph{e.g.} with dotted lines with no arrow
head), then it must be read from top to bottom.

\subsection{Interaction and composition}
\label{subsec:prelim_inter}

Consider two strategies $\sigma : A \stackrel{S}{\to} B$ and $\tau : B
\stackrel{T}{\to} C$. 

Recall that their \emph{interaction}
\[
\tau \inter \sigma : T \inter S \to A \parallel B \parallel C
\]
has set $\conf{T \inter S}$ isomorphic to pairs $(x^S, x^T) \in \conf{S}
\times \conf{T}$ such that $x^S_B = x^T_B = x_B$, and which are \emph{causally
compatible}, in the sense that the induced bijection 
\[
x^S \parallel x^T_C \simeq x^S_A \parallel x_B \parallel x^T_C \simeq x^S_A \parallel
x^T
\]
is secured \cite{cg2}. We write $x^T \inter x^S \in \conf{T\inter S}$ for the
corresponding configuration; then:
\[
(\tau \inter \sigma)(x^T \inter x^S) = x^S_A \parallel x_B \parallel x^T_C\,.
\]

The composition
\[
\tau \odot \sigma : A \stackrel{T\odot S}{\to} C
\]
is obtained from the interaction through a hiding operation \cite{cg2}. We recall:

\begin{proposition}
The set $\conf{T\odot S}$ is isomorphic to the set of pairs $(x^S, x^T) \in \conf{S}
\times \conf{T}$ such that $x^S_B = x^T_B = x_B$, which are causally compatible and
\emph{minimal}, in the sense that if $y^S \subseteq x^S$ and $y^T \subseteq x^T$ are
matching and causally compatible, and 
\[
x^S_A \parallel x^T_C = y^S_A \parallel y^T_C\,,
\]
then $x^S = y^S$ and $x^T = y^T$.
If $x^S$ and $x^T$ are matching, causally compatible, and minimal, we write $x^T
\odot x^S \in \conf{T\odot S}$ for the corresponding configuration. We then have
\[
(\tau \odot \sigma)(x^T \odot x^S) = x^S_A \parallel x^T_C\,.
\]
\end{proposition}
\begin{proof}
Direct from the definition. If a pair $(x^S, x^T)$ is matching and causally
compatible, then it is minimal iff $x^T \inter x^S$ has all its maximal events
visible (\emph{i.e.} in $A$ or $C$); and those are in one-to-one correspondence with
configurations of $T\odot S$.
\end{proof}

Interaction behaves like a cartesian product (restricted to the matching causally
compatible configurations), while composition has this additional minimality
assumption. We wish to get rid of minimality, since we wish to link to weighted
relational models, where (intuitively) a witness of the composition is a pair of
witnesses. This can be achieved:

\begin{definition}
Let $\sigma : S \to A$ be a strategy. 

A configuration $x \in \conf{S}$ is \emph{$+$-covered} iff all its maximal events
have positive polarity. We write $\pconf{S}$ for the set of $+$-covered
configurations of $\sigma$.
\end{definition}

By extension, we say that $x^T \inter x^S \in \conf{T\inter S}$ is $+$-covered iff
its maximal events are positive and write $x^T \inter x^S \in \pconf{T\inter S}$.
This notion is useful, because we have:

\begin{lemma}\label{lem:main_pcov}
Consider $\sigma : A \stackrel{S}{\to} B$ and $\tau : B \stackrel{T}{\to} C$ two
strategies. Then, there is a bijection
\[
\begin{array}{rcrcl}
\phi &:& \pconf{T\inter S} &\simeq& \pconf{T \odot S}\\
&& x^T \inter x^S &\mapsto& x^T \odot x^S
\end{array}
\]
such that if $(\tau \inter \sigma)(x^T \inter x^S) = x_A \parallel x_B \parallel x_C$,
then $(\tau \odot \sigma)(\phi(x^T \inter x^S)) = x_A \parallel x_C$.
\end{lemma}
\begin{proof}
If $x^T \inter x^S \in \pconf{T\inter S}$, then the pair $(x^S, x^T)$ is
automatically minimal: if not, then one can remove an event in $B$. But it must be
negative for either $\sigma$ or $\tau$, contradiction. So we may simply set $\phi(x^T
\inter x^S) = x^T \odot x^S \in \pconf{T\odot S}$.
\end{proof}

We have one last ingredient to introduce. One crucial difference between strategy
composition and composition in weighted relational models, is that strategies may
deadlock. This question is fairly well-explored; in particular in settings where we
have performed such a collapse \cite{lics18,popl20,mall}, we have done so under the
assumption that strategies satisfied a condition called \emph{visibility}, which
prevents deadlocks \cite{DBLP:phd/hal/Castellan17}. Describing visibility is beyond
the scope of this paper, but many of the results given here will be under the
assumption that certain strategies do not deadlock. Accordingly, we define: 

\begin{definition}\label{def:nodeadlock}
Strategies $\sigma : A \stackrel{S}{\to} B$ and $\tau : B \stackrel{T}{\to} C$ 
\emph{do not deadlock} iff for all 
$x^S \in \conf{S}$, $x^T \in \conf{T}$ and $\theta_B : x^S_B \sym_B x^T_B$, the
composite bijection
\[
x^S \parallel x^T_C 
\quad
\stackrel{\sigma \parallel x^T_C}{\simeq} 
\quad
x^S_A \parallel x^S_B \parallel x^T_C 
\quad
\stackrel{x^S_A \parallel \theta_B \parallel x^T_C}{\simeq}
\quad
x^S_A \parallel x^T_B \parallel x^T_C
\quad
\stackrel{x^S_A \parallel \tau^{-1}}{\simeq} 
\quad
x^S_A \parallel x^T
\]
is secured.
\end{definition}

This is, in particular, always the case when $\sigma$ and $\tau$ are visible. If
$\sigma$ and $\tau$ do not deadlock then we may forget the causal compatibility
condition in their interaction: configurations of the interaction correspond to
arbitrary matching pairs.

We do \emph{not} assume that all strategies considered do not deadlock. Throughout
the paper, we make it explicit when we consider this hypothesis.

\section{Towards a Quantitative Collapse}
\label{sec:quant_coll}

\subsection{Relational collapse and symmetry}

A game $A$ has a natural associated notion of \emph{position}, given
by the set of \emph{configurations} $\conf{A}$. Configurations inform the
relationship with relational-like semantics: if $A$ is a game arising from a type in
a \emph{linear type system}, then the \emph{web} (a set) interpreting this type in
relational semantics may be identified with a subset of
$\conf{A}$\footnote{Typically, in the presence of
Question/Answer labeling, those are the \emph{complete} configurations where every
question is answered -- but details do not matter for this paper.}. Likewise, a strategy 
\[
\sigma : A \stackrel{S}{\to} B
\]
induces a relation
$\coll \sigma = \{(x_A, x_B) \mid \exists x^S \in \conf{S}, \sigma x^S = x_A \parallel
x_B\} \in \Rel(\conf{A}, \conf{B})$.

With this definition, for any $\sigma : A \stackrel{S}{\to} B$ and $\tau
: B \stackrel{T}{\to} C$ we automatically have that
\[
\coll (\tau \odot \sigma) \subseteq (\coll \tau) \circ (\coll \sigma)
\]
and the other inclusion holds if $\sigma$ and $\tau$ do not deadlock.

This picture above is of course much simplified thanks to linearity. Without the
linearity assumption, the games considered need to carry a \emph{symmetry}. If $A$
arises from a type, then the corresponding web is no longer (a subset of) $\conf{A}$,
but (a subset of) $\sconf{A}$, the set of \emph{equivalence classes of configurations
under symmetry}. In particular, we have

\begin{lemma}
Consider $N$ a negative tcg. Then,
\[
\sconf{\oc N} \iso \mset{\sconf{N}}\,.
\]
where $\mset{X}$ is the set of \emph{finite multisets} of elements of set $X$.
\end{lemma}
\begin{proof}
Straightforward.
\end{proof}

We use $\x, \y,\dots $ as metavariables ranging over symmetry classes.

Above, $\oc$ stands for the AJM-style exponential described in Section 3.3.4 in
\cite{cg2}. Likewise, the reader familiar with relational semantics will recognize in
$\mset{X}$ the familiar exponential modality. This traces the path to
extend the links between game and relational semantics beyond the linear case: simply
correct the definition of $\coll \sigma$ by setting:
\[
\coll \sigma = \{(\x_A, \x_B) \in \sconf{A} \times \sconf{B} \mid \exists x^S \in
\conf{S}, x^S_A \in \x_A ~\&~x^S_B \in \x_B\}\,,
\]
where $\sigma x^S = x^S_A \parallel x^S_B$, a naming convention that we
shall adopt. If $x^S \in \conf{S}$ is such that $x^S_A \in
\x_A$ and $x^S_B \in \x_B$, we say that $x^S$ is a \textbf{witness} for $(\x_A,
\x_B)$ in $\coll \sigma$.

With this definition, it is immediate by
definition of composition of strategies that we retain $\coll(\tau \odot \sigma)
\subseteq \coll(\tau) \circ \coll(\sigma)$ for any strategies $\sigma : A
\stackrel{S}{\to} B$ and $\tau : B \stackrel{T}{\to} C$. 

\subsection{Synchronization up to symmetry}

More interesting is the reverse inclusion. Of course, the deadlock issue mentioned
above still applies. But something else is also going on: consider $\sigma : A
\stackrel{S}{\to} B$ and $\tau : B \stackrel{T}{\to} C$, and
\[
(\x_A, \x_B) \in \coll \sigma
\qquad
\qquad
(\x_B, \x_C) \in \coll \tau\,.
\] 

By definition, this means that there are $x^S \in \conf{S}$ and $x^T \in \conf{T}$
such that
\[
x^S_A \in \x_A,
\qquad
x^S_B \in \x_B\,,
\qquad
x^T_B \in \x_B\,,
\qquad
x^T_C \in \x_C\,.
\]

In particular, since we have $x^S_B \in \x_B$ and $x^T_B \in \x_B$ it follows that
there is a (non-unique)
\[
\theta : x^S_B \sym_B x^T_B\,,
\]
a symmetry on $B$. So the witnesses $x^S \in \conf{S}$ and $x^T \in \conf{T}$ might
not quite reach the same configuration of the game: typically, they might involve
completely distinct copy indices, and $\theta$ carries a reindexing from one to the
other. Independently of the deadlocks, if we wish to provide a witness $y \in
\conf{T\odot S}$ for $(\x_A, \x_C)$ in $\tau \odot \sigma$, we must in particular
find some $y^S \in \conf{S}$ and $y^T \in \conf{T}$ such that 
\[
y^S_A \in \x_A\,,
\qquad
y^S_B = y^T_B\,,
\qquad
y^T_C \in \x_C\,,
\]
matching on $B$ \emph{on the nose}. So starting from $x^S \in \conf{S}$ and $x^T \in
\conf{T}$, we must \emph{reindex} them until they match on $B$ on the nose.
Of course, this issue already arises in the process of constructing a game
semantics based on copy incides, to show that equivalence of (uniform)
strategies up to the choice of copy indices is stable under composition.

In thin concurrent games, the main tool to deal with it is the \emph{weak bipullback
property}:

\begin{lemma}[Weak bipullback property]\label{lem:weak_bipullback}
Let $\sigma : S \to A$ and $\tau : T \to A^\perp$ be pre-$\sim$-strategies. Let $x^S
\in \conf{S}$ and $x^T \in \conf{T}$ and $\theta : \sigma x^S \sym_A \tau x^T$, such
that the composite bijection
\[
x^S \stackrel{\sigma}{\simeq} \sigma x^S \stackrel{\theta}{\sym_A} \tau x^T
\stackrel{\tau}{\simeq} x^T
\]
is secured. Then, there are $y^S \in \conf{S}$ and $y^T \in \conf{T}$ causally
compatible, $\theta^S : x^S \sym_S y^S$ and $\theta^T : y^T \sym_T x^T$,
such that $\tau \theta^T \circ \sigma \theta^S = \theta$. Moreover, $y^S, y^T$ are
unique up to symmetry.
\end{lemma}

This appears as Lemma 3.23 in \cite{cg2}. The intuition is that $\sigma$ and $\tau$
play against each other, each replacing Player copy indices with one they are
prepared to play. By $\sim$-receptivity, $\tau$ must be receptive to a change in copy
indices made by Player, and reciprocally; so $y^S$ and $y^T$ may be constructed by
induction on the causal structure induced by the securedness assumption. If $\sigma :
A \stackrel{S}{\to} B$ and $\tau : B \stackrel{T}{\to} C$, and we have $x^S \in
\conf{S}$ and $x^T \in \conf{T}$ with
\[
\theta : x^S_B \sym_B x^T_B\,,
\]
we may apply the lemma above for $\sigma
\parallel C^\perp \to A^\perp \parallel B \parallel C^\perp$ and $A \parallel \tau :
A \parallel T \to A \parallel B^\perp \parallel C$. Provided some other argument
ensures the securedness assumption, then we obtain
\[
y^S \parallel y^T_C \in \conf{S \parallel C}
\qquad
\qquad
y^S_A \parallel y^T \in \conf{A \parallel T}
\]
matching on $B$; and so we have found an interaction
\[
y^T \inter y^S \in \conf{T\inter S}
\]
with $(\tau \inter \sigma)(y^T \inter y^S) = y^S_A \parallel y_B \parallel y^T_C$,
satisfying $y^S_A \in \x_A$ and $y^T_C \in \x_C$ thus providing through hiding the
desired witness for $(\x_A, \x_C) \in \coll (\tau \odot \sigma)$.

\subsection{Quantitative extension}

But the above is purely qualitative: if $\sigma : A \stackrel{S}{\to} B$ then the
collapse above lets us define which pairs $(\x_A, \x_B)$ are ``inhabited'' by
$\sigma$. This is sufficient in order to link game semantics with relational
semantics. But this is not sufficient if we want to reproduce this feat in the
presence of \emph{quantitative} information, such as probabilities or quantum
valuations.

For the purposes of this paper, let us say that we are now interested not in the
\emph{mere existence} of a witness $x^S \in \conf{S}$ such that $x^S_A \in \x_A$ and
$x^S_B \in \x_B$, but in \emph{counting} such witnesses. For reasons explained in
Section \ref{subsec:prelim_inter}, from now on we consider witnesses for $(\x_A,
\x_B)$ not merely those configurations $x^S \in \conf{S}$ such that $x^S_A \in \x_A$
and $x^S_B \in \x_B$; but those that are additionally $+$-covered, \emph{i.e.} we
have $x^S \in \pconf{S}$. 

From a strategy $\sigma : A \stackrel{S}{\to} B$, we want a
\emph{$\overline{\mathbb{N}}$}-weighted relation, \emph{i.e.} a function  
\[
\coll \sigma : \sconf{A} \times \sconf{B} \to \overline{\mathbb{N}}\,,
\]
where $\overline{\mathbb{N}} = \mathbb{N} \cup \{+\infty\}$, 
\emph{counting} the number of distinct witnesses for $(\x_A, \x_B)$. In that case,
for $\x_A \in \sconf{A}$ and $\x_B \in \sconf{B}$, write $(\coll \sigma)_{\x_A, \x_B}
\in \overline{\mathbb{N}}$ for the corresponding coefficient. 

In the spirit of weighted relations \cite{DBLP:conf/lics/LairdMMP13}, we then want to
prove that for all $\sigma : A\stackrel{S}{\to} B$ and $\tau : B \stackrel{T}{\to} C$
that do not deadlock, we have that for all $\x_A \in \sconf{A}$ and $\x_C \in
\sconf{C}$,   
\begin{eqnarray}
(\coll (\tau \odot \sigma))_{\x_A, \x_C} &=& \sum_{\x_B \in \sconf{B}} (\coll
\sigma)_{\x_A, \x_B} \times (\coll \tau)_{\x_B, \x_C}\,.\label{eq1}
\end{eqnarray}

The convergence of the sum on the right hand side is ensured by the fact
that we consider the completed natural numbers $\mathbb{N} \cup \{+\infty\}$ as in
the weighted relational model.

How might we, from $\sigma : A \stackrel{S}{\to} B$, extract the weighted
relation $\coll \sigma$? Intuitively, we need
\[
(\coll \sigma)_{\x_A, \x_B} = \card{\wit_\sigma(\x_A, \x_B)}
\]
where $\wit_\sigma(\x_A, \x_B)$ captures the \emph{witnesses} in $\sigma$ for
symmetry classes $\x_A \in \sconf{A}$ and $\x_B \in \sconf{B}$, and where $\card{X}$ 
simply computes the cardinal, taken to be $+\infty$ for $X$ infinite. Situations
where strategies carry additional weights, say probabilities or quantum valuations,
would be dealt with similarly. In any case, the first obstacle to overcome is then to
give a satisfactory definition of $\wit_\sigma(\x_A, \x_B)$. 

Of course counting all $x^S \in \pconf{S}$ such that $x^S_A \in \x_A$ and $x^S_B \in
\x_B$ makes no sense: there are almost always infinitely many of
them since \emph{e.g.} the construction $\oc N$ introduces countably many copy
indices. The definition of witnesses must take symmetry into account.

\paragraph{Symmetry classes.} The obvious candidate for witnesses, chosen in
\cite{lics18}, is:
\[
\wit_\sigma(\x_A, \x_B) = \{\x^S \in \spconf{S} \mid \forall x^S \in \x^S, x^S_A \in
\x_A ~\&~ x^S_B \in \x_B\}\,,
\]
\emph{i.e.} the symmetry classes of $+$-covered configurations mapping to
$\x_A \parallel \x_B$. This convincingly simple definition in fact hides a major
subtlety. Indeed, \eqref{eq1} hints at a bijection 
\[
\wit_{\tau \odot \sigma}(\x_A, \x_C) \iso \sum_{\x_B \in \sconf{B}}
\wit_\sigma(\x_A, \x_B) \times \wit_\tau(\x_B, \x_C)\,.
\]

This seems straightforward. 
Firstly, if $\z \in \wit_{\tau \odot \sigma}(\x_A, \x_C)$, then any choice $z
\in \z$ is $z = z^T \odot z^S \in \conf{T\odot S}$ and the symmetry classes
of its projections yield 
\[
\z^S \in \wit_\sigma(\x_A, \x_B)\,,
\qquad
\qquad
\z^T \in \wit_\tau(\x_B, \x_C)\,,
\]
for some $\x_B \in \sconf{B}$. These data are easily shown to be invariant under
the choice of $z$.

Reciprocally, if $\x^S \in \wit_\sigma(\x_A, \x_B)$ and $\x^T \in
\wit_\tau(\x_B, \x_C)$, we may take arbitrary $x^S \in \x^S, x^T \in \x^T$, and
via Lemma \ref{lem:weak_bipullback} find symmetric $y^S \in \x^S$ and $y^T \in \x^T$
agreeing on $B$ on the nose. We may then form $y^T \odot y^S \in \pconf{T\odot S}$ and
take its symmetry class in $\wit_{\tau \odot \sigma}(\x_A, \x_C)$. 

But one should not skip the details\footnote{We were guilty of that in
\cite{lics18}.}: we must show that this
construction only depends on the symmetry classes $\x^S$ and $\x^T$, not on the
specific choices $x^S \in \x^S$ and $x^T \in \x^T$ and the symmetry $\theta_B : x^S_B
\sym_B x^T_B$ used to link them. But surely, that must be true, right? 

Well, about that\dots~
It certainly was a surprise to us that the symmetry class
obtained through synchronization \emph{does} depend on the symmetry $\theta_B$.

\begin{example}\label{ex:ex1}
Consider the following games. Firstly, $A = \emptyset$ is the empty game. Secondly,
$C = (\oc \ominus)^\perp$ which has countably many Player moves written
$\done_{\grey{i}}$ for all $i \in \mathbb{N}$, all symmetric -- we adopt here a
convention followed throughout the paper: copy indices appear in grey, to distinguish
them from other indices.

Thirdly, consider the game $B = \oc_{HO}(\ominus \imc \oplus)$, where $\oc_{HO}$ is the
``HO exponential'' defined in Definition 2.24 with symmetries in Definition 2.27 in
\cite{cg2} (see also Proposition 3.3). This game has events, polarities and causal
dependency those pictured in: 
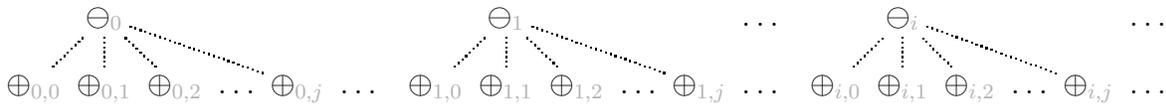
\[
\xymatrix@C=0pt@R=10pt{
&\ominus_{\grey{0}}
	\ar@{.}[dl]
	\ar@{.}[d]
	\ar@{.}[dr]
	\ar@{.}[drrr]
&&&&&&&
\ominus_{\grey{1}}
        \ar@{.}[dl]
        \ar@{.}[d]
        \ar@{.}[dr]
        \ar@{.}[drrr]
&&&&\dots&&&
\ominus_{\grey{i}}
        \ar@{.}[dl]
        \ar@{.}[d]
        \ar@{.}[dr]
        \ar@{.}[drrr]
&&&&\dots\\
\oplus_{\grey{0,0}}&
\oplus_{\grey{0,1}}&
\oplus_{\grey{0,2}}&\dots&
\oplus_{\grey{0,j}}&\dots&&
\oplus_{\grey{1,0}}&
\oplus_{\grey{1,1}}&
\oplus_{\grey{1,2}}&\dots&
\oplus_{\grey{1,j}}&\dots&&
\oplus_{\grey{i,0}}&
\oplus_{\grey{i,1}}&
\oplus_{\grey{i,2}}&\dots&
\oplus_{\grey{i,j}}&\dots
}
\]
with all finite sets consistent. Its symmetry comprises
all order-isomorphisms between configurations. Its positive symmetry comprises all
order-isomorphisms that preserve the initial (negative) move. Its negative symmetry
comprises all order-isomorphisms such that $\theta(\oplus_{\grey{i,j}}) =
\oplus_{\grey{i', j}}$ for some $i'\in \mathbb{N}$, \emph{i.e.} they preserve the $j$
component of the positive move. In practice, we will omit the first copy index for
the event in the second row, which is redundant with the immediate causal antecedent
of the event.

We now introduce two strategies $\sigma : A \stackrel{S}{\to} B$ and $\tau : B
\stackrel{T}{\to} C$ that we wish to compose, 
\begin{figure}
\begin{minipage}{.45\linewidth}
\[
\xymatrix@R=10pt@C=5pt{
&A&\stackrel{S}{\to}&B\\
&&&\ominus_{\grey{i}}
	\ar@{-|>}[dl]
	\ar@{-|>}[dr]\\
&&\oplus_{\grey{f(i)}}
	\ar@{.}@/^/[ur]
	\ar@{~}[rr]&&
\oplus_{\grey{g(i)}}
	\ar@{.}@/_/[ul]
}
\]
\caption{$\sigma : A \stackrel{S}{\to} B$}
\label{fig:sigma}
\end{minipage}
\hfill
\begin{minipage}{.45\linewidth}
\[
\xymatrix@R=10pt@C=0pt{
&B&&\stackrel{T}{\to} && C\\
\oplus_{\grey{0}}
	\ar@{-|>}[d]&&
\oplus_{\grey{h(i)}}
	\ar@{-|>}[d]&&&
\done_{\grey{k(i,j)}}\\
\ominus_{\grey{i}}
	\ar@{.}@/^/[u]
	\ar@{-|>}[urr]&&
\ominus_{\grey{j}}
	\ar@{-|>}[urrr]
	\ar@{.}@/^/[u]
}
\]
\caption{$\tau : B \stackrel{T}{\to} C$}
\label{fig:tau}
\end{minipage}
\end{figure}
represented on Figures \ref{fig:sigma} and \ref{fig:tau} where the functions $f, g,
h$, and $k$ are assumed injective, and $0$ is not in the codomain of $h$.
Note that the representation is symbolic: the diagrams must be understood by stating
that every positive move has one copy for each instantiation of the metavariables $i,
j \in \mathbb{N}$, with dependencies as indicated in the diagram. These copies are
compatible with each other. Finally, the symmetries comprise order-isomorphisms that
differ only by the value of the metavariables $i, j \in \mathbb{N}$. In particular,
the two moves in the conflicting branches of $\sigma$ are \emph{not} symmetric (that
would anyway contradict \emph{thinness}).

First, we compute the composition $\tau \odot \sigma$, and observe that it is:

\bigskip

\[
\xymatrix{
\done_{\grey{k(f(0), f(h(f(0))))}}
	\ar@{~}[r]
	\ar@{~}@/^1.5pc/[rr]
	\ar@{~}@/^2.5pc/[rrr]&
\done_{\grey{k(f(0), g(h(f(0))))}}
	\ar@{~}[r]
	\ar@{~}@/^1.5pc/[rr]&
\done_{\grey{k(g(0), f(h(g(0))))}}
	\ar@{~}[r]&
\done_{\grey{k(g(0), g(h(g(0))))}}
}
\]

There are four events, pairwise conflicting, reflecting the two non-deterministic
choices arising from the two calls to $\sigma$ -- one can read back which
non-deterministic choice gave rise to which result from the copy indices, but that is another
story. None of these events are symmetric: again, this would contradict thinness.

Now, let us define two configurations $x^S\in \conf{S}$ and $x^T \in \conf{T}$ as
\[
x^S = 
\raisebox{15pt}{$
\xymatrix@R=10pt@C=5pt{
\ominus_{\grey{0}}&
\ominus_{\grey{h(f(0))}}\\
\oplus_{\grey{f(0)}}
        \ar@{.}[u]&
\oplus_{\grey{g(h(f(0)))}}
        \ar@{.}[u]
}
$}
\qquad
\qquad
x^T = 
\raisebox{15pt}{$
\xymatrix@R=10pt@C=0pt{
\oplus_{\grey{0}}&
\oplus_{\grey{h(f(0))}}\\
\ominus_{\grey{f(0)}}
        \ar@{.}[u]&
\ominus_{\grey{g(h(f(0)))}}
        \ar@{.}[u]
}
$}
\parallel
\quad
\checkmark_{\grey{k(f(0), g(h(f(0))))}}
\]

These two configurations match on $B$ (and are causally compatible); and their
composition yields the configuration $\{\checkmark_{\grey{k(f(0), g(h(f(0))))}}\}$.
We of course obtain the same result if we synchonize them through the trivial
symmetry on their common interface:
\[
\id : x^T_B \sym_B x^T_B\,.
\]

But there is another endosymmetry on $x_B = x^S_B = x^T_B$, namely
\[
\swap : x^T_B \sym_B x^T_B\,,
\]
exhanging the two copies. Synchronizing $x^S$ and $x^T$ through $\swap$ via
Lemma \ref{lem:weak_bipullback} instead gives:
\[
\{\checkmark_{\grey{k(g(0), f(h(g(0))))}}\}
\]
which is not symmetric to $\{\checkmark_{\grey{k(f(0), g(h(f(0))))}}\}$ in $T\odot
S$. Indeed, intuitively, in $x^S$ we only have the information that there were two
calls to $\sigma$, with distinct non-deterministic resolutions. We do not know, just
by looking at $x^S$, which one is the ``first call'' and which one is the ``second
call''. The symmetry $\theta : x^S_B \sym_B x^T_B$ ``plugs'' the two calls in $x^T$
to their two non-deterministic resolutions in $x^S$. With $\id$ the first call
selected $\oplus_{\grey{f(i)}}$ and the second call $\oplus_{\grey{g(i)}}$, and the
other way around for $\swap$; leading to non-symmetric outcomes.
\end{example}

Well, this is puzzling. If the obvious candidate for a bijection between witnesses $z
\in \wit_{\tau \odot \sigma}(\x_A, \x_C)$ and pairs of witnesses $z^S \in
\wit_\sigma(\x_A, \x_B)$ and $z^T \in \wit_\tau(\x_B, \x_C)$ for some $\x_B$ does not
work, how can we hope to obtain \eqref{eq1}? This makes one wonder by what miracle
the weighted relational model works at all -- what does it really count?

\paragraph{Concrete witnesses.} To investigate this issue we introduce an
alternative, more concrete choice for witnesses. It is 
rooted in the following fact (Lemma 3.28 in \cite{cg2}):

\begin{lemma}\label{lemma:pos_symm}
Let $\sigma : S \to A$ be a pre-$\sim$-strategy on $A$, and
let $\theta : x \sym_{S} y$ such that $\sigma \theta \in {\sym_A^+}$.

Then, $x = y$ and $\theta = \id_x$.
\end{lemma}

For this, the condition \emph{thin} plays a crucial role. Intuitively,
\emph{thinness} means that the strategy has a canonical choice of copy indices for
its moves, once Opponent fixes their choice of copy indices. Accordingly, the lemma
above may be interpreted as saying that 
provided we remain in the positive symmetry (\emph{i.e.} we do not
change Opponent's copy indices), then the choice of the \emph{concrete} configuration
$x \in \conf{S}$ is unique. This suggests that we might take $\wit_\sigma(\x_A,
\x_B)$ to range over concrete configurations of $S$ matching with the game up to
\emph{positive} symmetry -- of course, for that we need reference concrete
configurations of the game rather than symmetry classes. So let us fix a choice, for
any tcg $A$ and any symmetry class $\x_A \in \sconf{A}$, of a concrete representative
written $\rep{\x}_A \in \x_A$.

Our alternative definition of witnesses is, for $\sigma : A
\stackrel{S}{\to} B$, $\x_A \in \sconf{A}$ and $\x_B \in \sconf{B}$:
\[
\wit_\sigma^+(\x_A, \x_B) = \{x^S \in \pconf{S} \mid x^S_A \sym_A^- \rep{\x}_A ~\&~
x^S_B \sym_B^+ \rep{\x}_B\}
\]

It will turn out (see Section \ref{sec:epilogue}) that (even assuming
representability) these two notions of witnesses are \emph{not} equivalent:
the weighted relational model counts not symmetry classes, but concrete witnesses up
to positive symmetry. In the rest of this paper, we aim to prove (as mentioned above,
modulo one additional condition on games) that $\wit^+$, unlike $\wit$, does the
trick\footnote{An early sign that $\wit^+$ is better behaved is that unlike $\wit$,
it does not depend on the choice of the symmetry for $\sigma$ -- recall from Section A.1.2 in
\cite{cg2} that the symmetry is \emph{not} unique.}.
This will be quite the ride, so switch off
your phone, fasten your seat belt, as we must now embark on a journey into the
darkest corners of thin concurrent games. 

\section{Canonical configurations and representable games}
\label{sec:canon_repr}

\subsection{Canonical representatives of symmetry classes}

To motivate the development of this section, let us look at the definition just
above:
\[
\wit_\sigma^+(\x_A, \x_B) = \{x^S \in \pconf{S} \mid x^S_A \sym_A^- \rep{\x}_A ~\&~
x^S_B \sym_B^+ \rep{\x}_B\}\,.
\]

This definition depends on a choice of a representative $\rep{\x}_A$, once and for
all, for every symmetry class $\x_A$. Of course, the set of witnesses we obtain this
way depends on this choice: a different choice of representatives yields
configurations of $S$ where Opponent uses different copy indices. 
But what we really need for this definition to be of any use, is that the
\emph{cardinal} of $\wit_\sigma^+(\x_A, \x_B)$ should not depend on the representatives
$\rep{\x}_A, \rep{\x}_B$.

Bad news: it does.

\begin{example}
Remember the game $B = \oc_{HO}(\ominus \imc \oplus)$ of Example \ref{ex:ex1}.
Consider the strategy 
\[
\xymatrix@R=10pt@C=10pt{
\oplus_{\grey{0}}\ar@{-|>}[d]&
\oplus_{\grey{h(i)}}
	\ar@{-|>}[d]\\
\ominus_{\grey{i}}
	\ar@{.}@/^/[u]
	\ar@{-|>}[ur]&
\ominus_{\grey{j}}
	\ar@{.}@/^/[u]
}
\]
written $\sigma : S \to B^\perp$, which is $\tau$ of Example \ref{ex:ex1} without the
last move. 

Now, imagine that we fix as representative for a symmetry class
in $B^\perp$ the configuration:
\[
\rep{\x}_B = 
\raisebox{12pt}{$
\xymatrix@R=10pt@C=10pt{
\oplus_{\grey{1}}&
\oplus_{\grey{2}}\\
\ominus_{\grey{1}}
	\ar@{.}[u]&
\ominus_{\grey{2}}
	\ar@{.}[u]
}$}\,.
\]

Let us consider the configurations of $S$ matching $\rep{\x}_B$ up to positive
symmetry. First, a
configuration $x \in \conf{S}$ matching our requirements has four moves, and each
Player move has \emph{exactly one} successor. So it must have the following form, for
some $i, j \in \mathbb{N}$,
\begin{eqnarray}&&
\raisebox{17pt}{$
\xymatrix@R=10pt@C=10pt{
\oplus_{\grey{1}}\ar@{-|>}[d]&
\oplus_{\grey{h(i)}}
        \ar@{-|>}[d]\\
\ominus_{\grey{i}}
        \ar@{.}@/^/[u]
        \ar@{-|>}[ur]&
\ominus_{\grey{j}}
        \ar@{.}@/^/[u]
}$}\label{eq2}
\end{eqnarray}
and finding the witnesses for $\rep{\x}_B$ boils down to figuring out all possible
positive symmetries
\[
\theta : 
\raisebox{17pt}{$
\xymatrix@R=10pt@C=10pt{
\oplus_{\grey{1}}&
\oplus_{\grey{h(i)}}\\
\ominus_{\grey{i}}
        \ar@{.}[u]&
\ominus_{\grey{j}}   
        \ar@{.}[u] 
}$}
\sym_{B^\perp}^+
\raisebox{17pt}{$
\xymatrix@R=10pt@C=10pt{
\oplus_{\grey{1}}&
\oplus_{\grey{2}}\\
\ominus_{\grey{1}}
        \ar@{.}[u]&
\ominus_{\grey{2}}
        \ar@{.}[u]
}$}
\]

The positive symmetry of $B^\perp$ is the negative symmetry of $B$: it lets
us change the indices of minimal events, but the second component for positive events 
must be left unchanged. We may freely associate the minimal events either as
$\oplus_{\grey{1}} \leftrightarrow \oplus_{\grey{1}}$ and $\oplus_{\grey{h(i)}}
\leftrightarrow \oplus_{\grey{2}}$; or as $\oplus_{\grey{1}} \leftrightarrow
\oplus_{\grey{2}}$ and $\oplus_{\grey{h(i)}} \leftrightarrow \oplus_{\grey{1}}$. But
if we do the former, as the symmetry is positive it forces $i = 1$ and $j=2$.
Likewise, if we do the latter, it forces $i=2$ and $j=1$. So overall, there are
\emph{exactly two} configurations of $S$ matching $\rep{\x}_B$ up to positive
symmetry: 
\[
\xymatrix@R=10pt@C=10pt{
\oplus_{\grey{1}}\ar@{-|>}[d]&
\oplus_{\grey{h(1)}}
	\ar@{-|>}[d]\\
\ominus_{\grey{1}}
	\ar@{.}@/^/[u]
	\ar@{-|>}[ur]&
\ominus_{\grey{2}}
	\ar@{.}@/^/[u]
}
\qquad\qquad
\xymatrix@R=10pt@C=10pt{
\oplus_{\grey{1}}\ar@{-|>}[d]&
\oplus_{\grey{h(2)}}
	\ar@{-|>}[d]\\
\ominus_{\grey{2}}
	\ar@{.}@/^/[u]
	\ar@{-|>}[ur]&
\ominus_{\grey{1}}
	\ar@{.}@/^/[u]
}
\]

In particular, there are \emph{two} witnesses for $\rep{\x}_B$. This is confusing,
because these two configurations are symmetric in $S$, so we seem to be counting the
same symmetry class of $S$ twice -- and we shall see indeed that this is a
pathological example. 

In contrast, assume we pick as representative for $\x_B$ the following configuration:
\[
\rep{\x}'_B = 
\raisebox{12pt}{$
\xymatrix@R=10pt@C=10pt{
\oplus_{\grey{1}}&
\oplus_{\grey{2}}\\
\ominus_{\grey{1}}
	\ar@{.}[u]&
\ominus_{\grey{1}}
	\ar@{.}[u]
}$}\,.
\]

Now, there is only \emph{exactly one} configuration of $S$ matching
$\rep{\x}'_B$ up to positive symmetry:
\[
\xymatrix@R=10pt@C=10pt{
\oplus_{\grey{1}}\ar@{-|>}[d]&
\oplus_{\grey{h(1)}}
	\ar@{-|>}[d]\\
\ominus_{\grey{1}}
	\ar@{.}@/^/[u]
	\ar@{-|>}[ur]&
\ominus_{\grey{1}}
	\ar@{.}@/^/[u]
}
\]

Indeed, starting from \eqref{eq2}, the positive symmetry forces $i$ and $j$ to be
both $1$; and we obtain the unique configuration above. So the choice of $\rep{\x}_B$
affects the number of witnesses.
\end{example}

What is the moral of the story? This is subtle. Notice that while there is indeed
\emph{exactly one} configuration  $x \in \conf{S}$ matching $\rep{\x}'_B$ up to
positive symmetry, there are still \emph{two symmetries} $\theta : \sigma x
\sym_{B^\perp}^+ \rep{\x}'_B$, corresponding to $\{\oplus_{\grey{1}} \leftrightarrow
\oplus_{\grey{1}}, \oplus_{\grey{h(1)}} \leftrightarrow \oplus_{\grey{2}}\}$ and
$\{\oplus_{\grey{1}} \leftrightarrow \oplus_{\grey{2}}, \oplus_{\grey{h(1)}}
\leftrightarrow \oplus_{\grey{1}}\}$. So for $\rep{\x}_B$ we get \emph{two}
witnesses, and each has \emph{one} positive symmetry to $\rep{\x}_B$; while for
$\rep{\x}'_B$ we get \emph{one} witness, with \emph{two} positive symmetries. So the
mismatch between the representatives is explained if one factors in the number of
positive symmetries.

To comment further: there are two positive endo-symmetries
$\rep{\x}'_B \sym_{B^\perp}^+ \rep{\x}'_B$: the identity, and the swap between positive
events. In contrast, in $\rep{\x}_B$, swapping the positive events
\[
\raisebox{12pt}{$
\xymatrix@R=10pt@C=10pt{
\oplus_{\grey{1}}&
\oplus_{\grey{2}}\\
\ominus_{\grey{1}}
        \ar@{.}[u]&
\ominus_{\grey{2}}
        \ar@{.}[u]
}$}
\quad
\sym_{B^\perp}^+
\quad
\raisebox{12pt}{$
\xymatrix@R=10pt@C=10pt{
\oplus_{\grey{2}}&
\oplus_{\grey{1}}\\
\ominus_{\grey{1}}
        \ar@{.}[u]&
\ominus_{\grey{2}}
        \ar@{.}[u]
}$}
\]
while preserving Opponent indices cannot be achieved via an endosymmetry, this
requires changing the configuration. To
avoid such pathological cases, we must select $\rep{\x}_B$ such that the positive
symmetry whose effect is, intuitively, merely to swap (the copy indices of) two Player
events, still has $\rep{\x}_B$ as codomain. We do not have a definition capturing
exactly this, as it is not clear how to formalize this idea of the minimal symmetry
``swapping two Player events''. However, for our purposes the following definition
does the job. 

\begin{definition}\label{def:canonical}
Consider $A$ a tcg, and $x \in \conf{A}$. 

We say that $x$ is \textbf{canonical} iff any $\theta : x \sym_A x$ factors
uniquely as
\[ 
x \stackrel{\theta^-}{\sym_A^-} x \stackrel{\theta_+}{\sym_A^+} x\,,
\]
with in particular $x$ in the middle.
\end{definition}

So endo-symmetries of canonical configurations decompose as
endo-symmetries, positive and negative. Of course we already know that all
endosymmetries (like all symmetries) decompose as the composite of a positive and
a negative symmetries (see Lemma 3.19 of \cite{cg2}). But there is a priori no reason
why the decomposition should have the same configuration in the middle.
This is in fact not always the case: for instance, picking the problematic
configuration $\rep{\x}_B$ of the example above, we have the decomposition  
\[
\raisebox{12pt}{$
\xymatrix@R=10pt@C=10pt{
\oplus_{\grey{1}}\ar@{.}[d]&
\oplus_{\grey{2}}\ar@{.}[d]\\
\ominus_{\grey{1}}&
\ominus_{\grey{2}}
}$}
\qquad
\sym_{B^\perp}^-
\qquad
\raisebox{12pt}{$
\xymatrix@R=10pt@C=10pt{
\oplus_{\grey{1}}\ar@{.}[d]&
\oplus_{\grey{2}}\ar@{.}[d]\\
\ominus_{\grey{2}}&
\ominus_{\grey{1}}
}$}
\qquad
\sym_{B^\perp}^+
\qquad
\raisebox{12pt}{$
\xymatrix@R=10pt@C=10pt{
\oplus_{\grey{2}}\ar@{.}[d]&
\oplus_{\grey{1}}\ar@{.}[d]\\
\ominus_{\grey{2}}&
\ominus_{\grey{1}}
}$}
\]
where rather than drawing the symmetries, we suggest them by considering that they
preserve the position of events in the diagrams. If we wish to avoid the problem
mentioned above, we must project strategies only on canonical representatives of
symmetry classes. But for that, we need to be sure that such canonical representatives
always exist.

Of course, there is no free lunch: in the full generality of tcgs, that is not the
case\footnote{The following example is due to Marc de Visme.}.

\begin{example}
Consider the tcg $A$, with events, polarities, and causality and follows:
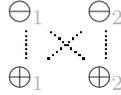
\[
\xymatrix@C=10pt@R=10pt{
\ominus_{\grey{1}}
	\ar@{.}[d]
	\ar@{.}[dr]&
\ominus_{\grey{2}}
	\ar@{.}[dl]
	\ar@{.}[d]\\
\oplus_{\grey{1}}&
\oplus_{\grey{2}}
}
\]

Its symmetry comprises all order-isomorphisms between configurations. The negative
symmetry has all order-isomorphisms included in one of the two maximal bijections
\[
\raisebox{15pt}{$
\xymatrix@R=10pt@C=10pt{
\ominus_{\grey{1}}
        \ar@{.}[d]
        \ar@{.}[dr]&
\ominus_{\grey{2}}
        \ar@{.}[dl]
        \ar@{.}[d]\\
\oplus_{\grey{1}}&
\oplus_{\grey{2}}
}$}
\sym^-_A
\raisebox{15pt}{$
\xymatrix@R=10pt@C=10pt{
\ominus_{\grey{1}}
        \ar@{.}[d]
        \ar@{.}[dr]&
\ominus_{\grey{2}}
        \ar@{.}[dl]
        \ar@{.}[d]\\
\oplus_{\grey{1}}&
\oplus_{\grey{2}}
}$}
\qquad
\qquad
\raisebox{15pt}{$
\xymatrix@R=10pt@C=10pt{
\ominus_{\grey{1}}
        \ar@{.}[d]
        \ar@{.}[dr]&
\ominus_{\grey{2}}
        \ar@{.}[dl]
        \ar@{.}[d]\\
\oplus_{\grey{1}}&
\oplus_{\grey{2}}
}$}
\sym^-_A
\raisebox{15pt}{$
\xymatrix@R=10pt@C=10pt{
\ominus_{\grey{2}}
        \ar@{.}[d]
        \ar@{.}[dr]&
\ominus_{\grey{1}}
        \ar@{.}[dl]
        \ar@{.}[d]\\
\oplus_{\grey{2}}&
\oplus_{\grey{1}}
}$}
\]
where again, the bijection matches those events in the corresponding position of the
diagram. Likewise, the positive symmetry has all order-isomorphisms included in one
of:

\[
\raisebox{15pt}{$
\xymatrix@R=10pt@C=10pt{
\ominus_{\grey{1}}
        \ar@{.}[d]
        \ar@{.}[dr]&
\ominus_{\grey{2}}
        \ar@{.}[dl]
        \ar@{.}[d]\\
\oplus_{\grey{1}}&
\oplus_{\grey{2}}
}$}
\sym^+_A
\raisebox{15pt}{$
\xymatrix@R=10pt@C=10pt{
\ominus_{\grey{1}}
        \ar@{.}[d]
        \ar@{.}[dr]&
\ominus_{\grey{2}}
        \ar@{.}[dl]
        \ar@{.}[d]\\
\oplus_{\grey{1}}&
\oplus_{\grey{2}}
}$}
\qquad
\qquad
\raisebox{15pt}{$
\xymatrix@R=10pt@C=10pt{
\ominus_{\grey{1}}
        \ar@{.}[d]
        \ar@{.}[dr]&
\ominus_{\grey{2}}
        \ar@{.}[dl]
        \ar@{.}[d]\\
\oplus_{\grey{1}}&
\oplus_{\grey{2}}
}$}
\sym^+_A
\raisebox{15pt}{$
\xymatrix@R=10pt@C=10pt{
\ominus_{\grey{1}}
        \ar@{.}[d]
        \ar@{.}[dr]&
\ominus_{\grey{2}}
        \ar@{.}[dl]
        \ar@{.}[d]\\
\oplus_{\grey{2}}&
\oplus_{\grey{1}}
}$}
\]
forming, altogether, a tcg. Then, the endosymmetry
\[
\raisebox{15pt}{$
\xymatrix@R=10pt@C=0pt{
\ominus_{\grey{1}}
	\ar@{.}[dr]&&
\ominus_{\grey{2}}
	\ar@{.}[dl]\\
&\oplus_{\grey{1}}
}$}
\qquad
\sym_A
\qquad
\raisebox{15pt}{$
\xymatrix@R=10pt@C=0pt{
\ominus_{\grey{2}}
        \ar@{.}[dr]&&
\ominus_{\grey{1}}
        \ar@{.}[dl]\\
&\oplus_{\grey{1}}
}
$}
\]
which is neither positive nor negative, uniquely factors as
\[
\raisebox{15pt}{$
\xymatrix@R=10pt@C=0pt{
\ominus_{\grey{1}}
        \ar@{.}[dr]&&
\ominus_{\grey{2}}
        \ar@{.}[dl]\\
&\oplus_{\grey{1}}
}$}
\qquad
\sym_A^-
\qquad
\raisebox{15pt}{$
\xymatrix@R=10pt@C=0pt{
\ominus_{\grey{2}}
        \ar@{.}[dr]&&
\ominus_{\grey{1}}
        \ar@{.}[dl]\\
&\oplus_{\grey{2}}
}
$}
\qquad
\sym_A^+
\qquad
\raisebox{15pt}{$
\xymatrix@R=10pt@C=0pt{
\ominus_{\grey{2}}
        \ar@{.}[dr]&&
\ominus_{\grey{1}}
        \ar@{.}[dl]\\
&\oplus_{\grey{1}}
}
$}
\]
which is not formed of endosymmetries. So this configuration is not canonical, but
its only symmetric $\{\ominus_{\grey{1}}, \ominus_{\grey{2}}, \oplus_{\grey{2}}\}$ is
not canonical either, for  the same reason.
\end{example}

Fortunately, no such pathological example arises in the games that (to our
knowledge) have found a use in semantics of logics and programming languages. Next we shall
propose the existence of a canonical representative as a new axiom for tcgs, and show
that it is preserved by all useful constructions on games.

\subsection{Representable games}

The axiom of 
\emph{representability} simply requires the existence of canonical representatives. 

\begin{definition}
Consider $A$ a tcg. 

We say that $A$ is \textbf{representable} iff for all $\x \in \sconf{A}$, there is
$\rep{\x} \in \x$ canonical.
\end{definition}

If $A$ is representable we may consider fixed in advance a choice, for every symmetry
class $\x \in \sconf{A}$, of a canonical representative $\rep{\x} \in \conf{A}$. 
For this to be a reasonable condition on tcgs, we must check that all the common
game constructions preserve representability.

\paragraph{Basic constructions.} First, we review the common game constructions that
have few interactions with the symmetry. Clearly, the empty game is representable. We
have:

\begin{lemma}
Consider $A, B$ representable tcgs. Then,
\[
\begin{array}{ll}
\text{\emph{(1)}} & \text{$A^\perp$ is representable,}\\
\text{\emph{(2)}} & \text{$A\parallel B$ is representable.}
\end{array}
\]
\end{lemma}
\begin{proof}
\emph{(1)} the dual exchanges $\sym_A^+$ and $\sym_A^-$ and the
definition of canonical is symmetric.

\emph{(2)} If $\x_A \parallel \x_B \in \sconf{A\parallel B}$, we simply set 
$\rep{\x_A \parallel \x_B} = \rep{\x}_A \parallel \rep{\x}_B$. Canonicity follows
directly from that of $\rep{\x}_A$ and $\rep{\x}_B$, exploiting the fact that any
endosymmetry  
\[
\theta \quad : \quad \rep{\x}_A \parallel \rep{\x}_B \quad \sym_{A\parallel B} \quad
\rep{\x}_A \parallel \rep{\x}_B 
\]
must have the form $\theta = \theta_A \parallel \theta_B$ for endosymmetries
$\theta_A : \rep{\x}_A \sym_A \rep{\x}_A$ and $\theta_B : \rep{\x}_B \sym_B
\rep{\x}_B$.
\end{proof}

The above are the game constructions used in the compact closed structure of thin
concurrent games. With similarly direct proofs, we cover all the frequent
constructions on tcgs that are essentially independent of symmetry: the shifts
$\uparrow A$  (resp. $\downarrow A$) which prefix the game $A$ with a new negative
(resp. positive) move (see \emph{e.g.} \cite{popl20}), the \emph{sum} $\sum_{i\in I}
A_i$ having all $A_i$ in pairwise conflict (see \emph{e.g.} \cite{popl20}), the
linear arrow $M \lin N$ of negative $M, N$ -- for all those, preservation of
representability is direct. What requires more care is the fact that the
constructions that introduce symmetry do indeed preserve representability. 

\paragraph{HO exponential.} We start with the Hyland-Ong style exponential. Recall
that it takes an \emph{arena} in the usual Hyland-Ong sense, \emph{i.e.} a forestial
partial order, without symmetry. We refer to \cite{cg2} for the definition of $\oc_{HO}
A$ for $A$ an arena and the associated notations.

We have the proposition:

\begin{proposition}
For $A$ any arena, $\oc_{HO} A$ is a representable thin concurrent game.
\end{proposition}
\begin{proof}
Within this proof (and only), by $\oc A$ we mean $\oc_{HO} A$.
Those configurations $x\in \conf{\oc_{HO} A}$ with exactly
one initial move are entirely determined by:
\[
\begin{array}{ll}
\text{\emph{(1)}} & \text{their \emph{label} $\lbl(\min(x))$,}\\
\text{\emph{(2)}} & \text{their copy index $\ind(\min(x))$,}\\
\text{\emph{(3)}} & \text{for each $\min(x) \imc a$, the sub-configuration starting with $a$.}
\end{array}
\]

Any $x \in \conf{\oc A}$ with index $i$, label $a$ and
sub-configurations $x_1, \dots, x_n$ may be written
\[
x = i \cdot (\{x_1, \dots, x_n\} \lin a) \in \conf{\oc A}
\]
where each $x_i$ is written similarly, with a notation inspired from intersection
types. But then, using that similarly any $x_j$ is written $i_j \cdot (X_j \lin
a_j)$, we may rewrite $x$ as 
\[
i \cdot ((i_1\cdot x'_1, \dots, i_n \cdot x'_n) \lin a)
\]
where each $x'_j = X_j \lin a_j$. Going one step further, write
\[
x = i \cdot ((i^1_1 \cdot x^1_1, \dots, i^1_{p_1} \cdot x^1_{p_1}) \lin \dots \lin
(i^m_1 \cdot x^m_1, \dots, i^m_{p_m} \cdot x^m_{p_m}) \lin a)\,,
\]
regrouping sub-trees by symmetry classes. If $x$ is to be canonical, then for any
$1\leq k \leq m$, any $i^k_l$ and $i^k_{l'}$ should be swapped by an endosymmetry;
implying $x^k_l = x^k_{l'}$. So we set
\[
x' = i \cdot ((i^1_1 \cdot x^1_1, \dots, i^1_{p_1} \cdot x^1_1) \lin \dots \lin
(i^m_1 \cdot x^m_1, \dots, i^m_{p_m} \cdot x^m_1) \lin a)\,,
\]
which is symmetric to $x$ by construction; moreover if for all $1\leq k \leq m$,
$x^k_1$ is assumed canonical by induction hypothesis, then one may verify that $x'$
is canonical.
\end{proof}

We omit the details on that last verification, as it is the exact same reasoning as
for the AJM exponential, which we give more formally below.  Unsuprisingly, the proof
for AJM bears much in common with the one above. 
We started with HO as we believe that the more concrete nature of games obtained
through the HO exponential makes the
reasoning slightly more transparent: 
we wish to construct a configuration where any two moves with swappable
copy indices have the \emph{exact same} sub-trees below, so that the two copy indices
may be simply swapped leaving the remainder of the configuration unchanged.

\paragraph{AJM exponential.} The AJM exponential is our main 
source of non-trivial symmetries.

\begin{lemma}
Consider $N$ a representable \emph{negative} thin concurrent game, \emph{i.e.} all
its minimal events are negative.
Then, the thin concurrent game $\oc N$ is representable.
\end{lemma}
\begin{proof}
Let $x \in \conf{\oc N}$, of the form $x =\,\parallel_{i\in I}
x_i$, where $x_i \in \conf{N}$. Let us partition $I$ as
\[
I = \biguplus_{k\in K} I_k
\]
such that for all $i, j \in I$, $x_i \sym_N x_j$ iff there is some $k\in
\mathbb{N}$ such that $i, j \in K$. For each $i\in I$, write $f(i) \in K$ for the
corresponding component. For each $k \in K$, fix some $g(k) \in I_k$. 

Now, fix $k \in K$. Since $N$ is representable, there is $x_{g(k)} \sym_N
\canon(x_{g(k)})$ with $\canon(x_{g(k)})$ canonical. Then for each $j \in I_k$
we replace $x_j$ with $\canon(x_{g(k)})$; or more formally we set
\[
x' =\,\parallel_{i\in I} \canon(x_{g(f(i))}) \in \conf{\oc N}\,.
\]

We clearly have $x \sym_{\oc N} x'$; indeed, for each $i\in I$, we have $x_i \sym_N
x_{g(f(i))} \sym_N \canon(x_{g(f(i))})$. Furthermore, $x'$ is canonical. Indeed,
writing $x'_i = \canon(x_{g(f(i))})$, consider now any symetry
\[
\theta ~~ : ~~ \parallel_{i\in I} x'_i ~~ \sym_{\oc N} ~~ \parallel_{i\in I} x'_i\,.
\]

By definition, there is $\pi : I \to I$ a permutation, and for all $i\in I$ a
symmetry $\theta_i : x'_i \sym_N x'_{\pi(i)}$. But by construction, this means that
we had $x_i \sym_N x_{\pi(i)}$ as well, so $i, \pi(i)$ belong to the same component
of the partition and $g(f(i)) = g(f(\pi(i)))$. Therefore, by construction, $x'_i =
x'_{\pi(i)}$. But $x'_i$ is canonical, so $\theta_i$ decomposes as 
\[
x'_i ~~ \stackrel{\theta_i^-}{\sym_N^-} ~~ x'_i ~~ \stackrel{\theta_i^+}{\sym_N^+} ~~
x'_i\,.
\]

Setting $\theta^-(i,e) = (\pi(i), \theta_i^-(e))$ and $\theta^+(i,e) = (i,
\theta^+_i(e))$, we have the required decomposition of $\theta$, showing that $x'$ is
canonical, as required.
\end{proof}

In the rest of this paper, we aim to make it explicit whenever this condition is
required.

\section{Quantitative collapse}
\label{sec:quant_coll_final}

By now we have added a new condition on games which eliminates some pathological
examples, and we have proved that this condition is preserved by all sensible
constructions on games. It remains to be seen whether this condition does solve the
problem at hand. 

\subsection{Actions of negative symmetries on strategies}
\label{subsec:neg_act}

Before we start, recall that for $A, B$ tcgs (which from now on will always be
assumed to be representable), $\sigma : A \stackrel{S}{\to} B$ a strategy and $\x_A
\in \sconf{A}$, $\x_B \in \sconf{B}$, we have set
\[
\wit_\sigma^+(\x_A, \x_B) = \{x^S \in \pconf{S} \mid x^S_A \sym_A^- \rep{\x}_A
~\&~x^S_B \sym_B^+ \rep{\x}_B\}\,,
\]
where $\rep{\x}_A$ and $\rep{\x}_B$ are the canonical representatives given by
representability of $A$ and $B$.

Our next step will be to investigate how negative symmetries act on witnesses. Our
starting point for that is the following lemma, Lemma B.4 in \cite{cg2}.

\begin{lemma}\label{lem:b4}
Consider $\sigma : S \to A$ a pre-$\sim$-strategy, $x^S \in \conf{S}$ and $\theta_- :
x^S_A \sym_A^- y_A$.
Then, there is a \emph{unique} $\varphi : x^S \sym_S y^S$ s.t. $\sigma \varphi =
\theta_+ \circ \theta_- : x^S_A \sym_A y^S_A$ for some $\theta_+ : y_A \sym_{A}^+
y^S_A$.
\end{lemma}

This is our main tool to have negative symmetries act on strategies. 
If $x^S \in \conf{S}$ and $\theta_- : x^S_A \sym_A^- y_A$ presents
a change in Opponent's copy indices, we can make $\theta_-$ ``act on''
$x^S$: Player adapts to the change of Opponent copy indices and presents some
$\varphi : x^S \sym_S y^S$. 

It is tempting to invoke some group theory here. For any $x \in
\conf{A}$, we have three groups: the group $\ssym{x}$ of endosymmetries $\theta : x
\sym_A x$, the group $\psym{x}$ of \emph{positive} endosymmetries, and the group
$\nsym{x}$ of \emph{negative} endosymmetries. When applied to symmetry classes, as in
$\ssym{\x}$ for $\x \in \sconf{A}$, these operations mean $\ssym{\rep{\x}}$. Of
course, if $x \sym_A y$ then any $\theta : x \sym_A y$ provides an iso
between $\ssym{x}$ and $\ssym{y}$ by conjugation. Warning: if $x \sym_A y$ we
\emph{do not} necessarily have $\nsym{x}$ and $\nsym{y}$ isomorphic (and of course,
likewise for $\psym{-}$), so the notation $\nsym{\x}$ is borderline -- we insist that
it means $\nsym{\rep{\x}}$ and depends on the chosen representative. This shall
hopefully cause no confusion.

Now, for $\sigma : S \to A$ and $x_A \in \conf{A}$, it is
tempting to make $\nsym{x_A}$ act on the set
\[
X = \{x^S \in \conf{S} \mid \sigma x^S = x_A\}\,,
\]
but for $\theta_- \in \nsym{x_A}$ and $x^S \in X$, there is no
reason why the $\varphi : x^S \sym_S y^S$ obtained via Lemma \ref{lem:b4} would
satisfy $\sigma y^S = x_A$ and hence remain in $X$.

So we add a bit of wiggling room. For $\x_A \in \sconf{A}$, we define the set
\[
\pswit^+(\x_A) = \{(x^S, \theta_+) \mid x^S \in \pconf{S},~\theta_+ : x^S_A \sym_A^+ \rep{\x}_A \}
\]
of witnesses for $\x_A$ along with a specific choice of positive symmetry. Then we
indeed have:

\begin{proposition}\label{prop:act1}
Consider $A$ a tcg and $\x_A \in \sconf{A}$. There is a group action
\[
(\_ \acts \_) : \nsym{\x_A} \times \pswit^+(\x_A) \to \pswit^+(\x_A)\,,
\]
such that for all $(y^S, \psi_+) = \varphi_- \acts (x^S, \theta_+)$, there is $\phi^S :
x^S \sym_S y^S$ making the diagram
\[
\xymatrix{
x^S_A   \ar[rr]^{\theta_+}
        \ar[d]_{\phi^S_A}&&
\rep{\x}_A
        \ar[d]^{\varphi_-}\\
y^S_A   \ar[rr]_{\psi_+}&&
\rep{\x}_A
}
\]
commute.
\end{proposition}
\begin{proof}
Consider $(x^S, \theta_+) \in \pswit^+(\x_A)$ and $\varphi_- \in \nsym{\x_A}$. We show
that there is unique $\phi^S : x^S \sym_S y^S$ and $\psi_+ : y^S_A \sym_A^+
\rep{\x}_A$ making the following diagram commute:
\[
\xymatrix{
x^S_A	\ar[rr]^{\theta_+}
	\ar[d]_{\phi^S_A}&&
\rep{\x}_A
	\ar[d]^{\varphi_-}\\
y^S_A	\ar[rr]_{\psi_+}&&
\rep{\x}_A
}
\]

For existence, by Lemma 3.19 of \cite{cg2}, $\varphi_- \circ \theta_+ : x^S_A \sym_A
\rep{\x}_A$ factors uniquely as 
\[
\Xi_+ \circ \Xi_- : x^S_A \sym_A \rep{\x}_A\,.
\]

Next, by Lemma \ref{lem:b4}, there is $\phi^S : x^S \sym_S y^S$ such that we have
\[
\phi^S_A = \Omega_+ \circ \Xi_- : x^S_A \sym_A y^S_A
\]
for some $\Omega_+ : y_A \sym_A^+ y^S_A$. We then form $\psi_+ = \Xi_+ \circ
\Omega_+^{-1}$ to conclude.

For uniqueness, if we have $\varphi_1 : x^S \sym_S y^S$ and $\varphi_2 : x^S \sym_S
z^S$ satisfying the requirements, 
\[
\xymatrix@C=40pt{
y^S_A	\ar[r]^{(\sigma \varphi_1)^{-1}}
	\ar[d]_+&
x^S_A	\ar[d]_+
	\ar[r]^{\sigma\varphi_2}&
z^S_A	\ar[d]^+\\
\rep{\x}_A
	\ar[r]_{\varphi_-^{-1}}&
\rep{\x}_A
	\ar[r]_{\varphi_-}&
\rep{\x}_A
}
\]
commutes, so $(\sigma \varphi_2) \circ (\sigma \varphi_1)^{-1} = \sigma (\varphi_2
\circ \varphi_1^{-1})$ is positive, so by Lemma 3.28 of \cite{cg2} we have $\varphi_2
\circ \varphi_1^{-1} = \id$, so $\varphi_1 = \varphi_2$.
\end{proof}

Note that in the proof, we have actually not used the representability assumption.
However, it will come in to deduce a property useful for elaborate forms of the
collapse (namely, in the quantum case). For that, we need the following intermediate
lemma.

\begin{lemma}\label{lem:dec_canonical}
Consider $A$ a representable tcg, $\x_A \in \sconf{A}$, and $x \in \conf{A}$ s.t.
$x \sym_A^+ \rep{\x}_A$.

Then, any $\theta : x \sym_A \rep{\x}_A$ factors uniquely as $\theta_- \circ
\theta_+$, where $\theta_+ : x \sym_A^+ \rep{\x}_A$ and $\theta_- \in
\nsym{\rep{\x}_A}$. 
\end{lemma}
\begin{proof}
Fix some $\varphi : x \sym_A^+ \rep{\x}_A$. Now, take $\theta : x \sym_A \rep{\x}_A$.
By Lemma 3.19 of \cite{cg2}, $\theta$ factors uniquely as $\theta_- \circ \theta_+$,
where $\theta_+ : x \sym_A^+ z$ and $\theta_- : z \sym_A^- \rep{\x}_A$ for some $z
\in \conf{A}$. But then, 
\[
\varphi \circ \theta^{-1} : \rep{\x}_A \sym_A \rep{\x}_A
\]
factors via $(\varphi \circ \theta_+^{-1}) : z \sym_A^+ \rep{\x}_A$ and
$\theta_-^{-1}  : \rep{\x}_A \sym_A^- z$, so $\rep{\x}_A = z$ follows since
$\rep{\x}_A$ is canonical. 
\end{proof}

For $A$ a tcg and $\x_A \in \sconf{A}$, we have previously defined
\[
\pswit^+(\x_A) = \{(x^S, \theta_+) \mid x^S \in \pconf{S},~\theta_+ : x^S_A \sym_A^+ \rep{\x}_A \}
\]
the set of witnesses for $\rep{\x}_A$ up to positive symmetry, \emph{along with} a
specific choice of positive symmetry $\theta_+ : x^S_A \sym_A^+ \rep{\x}_A$. We shall
now also consider the variation
\[
\swit^+(\x_A) = \{(x^S, \theta) \mid x^S_A \sym_A^+ \rep{\x}_A ~\&~ \theta : x^S_A
\sym_A \rep{\x}_A \}
\]
where we know that $x^S_A \sym_A^+ \rep{\x}_A$, but $\theta : x^S_A
\sym_A \rep{\x}_A$ may not be positive.

\begin{corollary}
Consider $A$ a representable tcg and $\x_A \in \sconf{A}$. Then, the function
\[
\begin{array}{rcrcl}
F&:&\swit^+(\x_A) &\to& \pswit^+(\x_A)\\
&&(x^S, \theta_- \circ \theta_+) &\mapsto& \theta_- \acts (x^S, \theta_+)
\end{array}
\]
is such that any $X \in \pswit^+(\x_A)$ has exactly $\card{\nsym{\rep{\x}_A}}$
antecedents.
\end{corollary}
\begin{proof}
The definition of $F$ makes use of the decomposition of all symmetries $\theta :
x^S_A \sym_A \rep{\x}_A$ offered by Lemma \ref{lem:dec_canonical}, using canonicity
of $\rep{\x}_A$. The statement on the number of antecedents is an immediate
consequence of the group action of Proposition \ref{prop:act1}.
\end{proof}

The reader might not immediately see the point; in fact we will not use this to
establish \eqref{eq1}, but it fits in this paper as it is required for more elaborate
versions of this construction, in particular in the presence of quantum valuations
\cite{popl20}.

\subsection{Quantitative synchronization up to symmetry}

Let us fix for this section two strategies $\sigma : A \stackrel{S}{\to}
B$ and $\tau : B \stackrel{T}{\to} C$.

\paragraph{Witnessing strategies and interactions.} 
We write elements of $\pswit_\sigma^+(\x_A, \x_B)$ as triples $(\theta^A_-, x^S,
\theta^B_+)$; as an alias for $(x^S, \theta^A_- \parallel \theta^B_+) \in
\pswit^+_\sigma(\x_A \parallel \x_B)$. Two witnesses
\[
(\theta^A_-, x^S, \theta^B_+) \in \pswit_\sigma^+(\x_A, \x_B)\,,
\qquad
(\Omega^B_-, x^T, \Omega^C_+) \in \pswit_\tau^+(\x_B, \x_C)\,,
\]
are causally compatible iff the composite bijection (see Definition
\ref{def:nodeadlock}) is secured. We write 
\[
\pswit_\sigma^+(\x_A, \x_B) \bullet \pswit_\tau^+(\x_B, \x_C)
\]
for the set of causally compatible pairs $(\w_\sigma, \w_\tau) \in \pswit_\sigma^+(\x_A,
\x_B) \times \pswit_\tau^+(\x_B, \x_C)$.

To accompany our notions of witnesses for strategies we
shall need to provide witnesses for interactions. If $\x_A \in \sconf{A}$, $\x_B \in
\sconf{B}$ and $\x_C \in \sconf{C}$, we write
\[
\wint^+_{\tau \inter \sigma}(\x_A, \x_B, \x_C) = 
\{x^T \inter x^S \in \pconf{T\inter S} \mid x^S_A \sym_A^- \rep{\x}_A,~x^S_B = x^T_B
\sym_B \rep{\x}_B,~\&~x^T_C \sym_C^+ \rep{\x}_C\}\,.
\]

Like for strategies, we also write $\pswint^+_{\tau \inter \sigma}(\x_A, \x_B, \x_C)$
for the set
\[
\{(\theta_-^A, x^T\inter x^S, \theta_+^C) \mid \theta_-^A : x^S_A \sym_A^-
\rep{\x}_A,~x^T \inter x^S \in \wint^+_{\tau \inter \sigma}(\x_A, \x_B,
\x_C),~\&~\theta_+^C : x^T_C \sym_C^+ \rep{\x}_C\}\,.
\]
interaction witnesses along with specific symmetries to the game. Finally, we write:
\[
\wint^+_{\tau\inter \sigma}(\x_A, \x_C) = 
\{x^T \inter x^S \in \pconf{T\inter S} \mid x^S_A \sym_A^- \rep{\x}_A,~\&~x^T_C \sym_C^+ \rep{\x}_C\}
\]
for the variant of $\wint^+_{\tau \inter \sigma}(\x_A, \x_B, \x_C)$ with no
constraint in $B$. Clearly, we have:

\begin{lemma}\label{lem:part_int}
Consider $\sigma : A \stackrel{S}{\to} B$ and $\tau : B \stackrel{T}{\to} C$, $\x_A
\in \sconf{A}$ and $\x_C \in \sconf{C}$. Then:
\[
\wint^+_{\tau \inter \sigma}(\x_A, \x_C) = 
\biguplus_{\x_B \in \sconf{B}} \wint^+_{\tau \inter \sigma}(\x_A, \x_B, \x_C)
\]
where the notation $\uplus$ means the plain set-theoretic union when it is
disjoint.
\end{lemma}
\begin{proof}
Simply partition interactions according to the symmetry class reached in $B$.
\end{proof}

\paragraph{Interactions up to symmetry.}
We start with a more explicit variant of Lemma \ref{lem:weak_bipullback}.

\begin{lemma}\label{lem:qbip1}
For any pair of causally compatible witnesses
\[
(\theta_-^A, x^S, \theta_+^B) \in \pswit_\sigma^+(\x_A, \x_B)\,,
\qquad
(\Omega^B_-, x^T, \Omega^C_+) \in \pswit_\tau^+(\x_B, \x_C)\,,
\]
there are unique symmetries $\omega^S : x^S \sym_S y^S, \nu^T : x^T \sym_T y^T$, 
$\Theta_B : \rep{\x}_B \sym_B y_B$ and witness
\[
(\psi_-^A, y^T \inter y^S, \psi_+^C) \in \pswint_{\tau \inter \sigma}^+(\x_A,
\x_B, \x_C)
\]
with $y^S_B = y^T_B = y_B$, such that the following diagrams commute:
\[
\xymatrix@R=10pt{
&x^S_A	\ar[dl]_{\theta^A_-}
	\ar[dd]^{\omega^S_A}&
x^S_B	\ar[r]^{\theta^B_+}
	\ar[dd]_{\omega^S_B}&
\rep{\x}_B	
	\ar[dd]^{\Theta_B}&
x^T_B	\ar[dd]^{\nu^T_B}
	\ar[l]_{\Omega^B_-}&
x^T_C	\ar[dr]^{\Omega^C_+}
	\ar[dd]_{\nu^T_C}\\
\rep{\x}_A&&&&&&\rep{\x}_C\\
&y^S_A	\ar[ul]^{\psi^A_-}&
y^S_B 	\ar@{=}[r]&
y_B&
y^T_B 	\ar@{=}[l]&
y^T_C 	\ar[ur]_{\psi^C_+}
}
\]
\end{lemma}
\begin{proof}
By Lemma \ref{lem:weak_bipullback}, there are unique symmetries $\omega^S : x^S
\sym_S y^S, \nu^T : x^T \sym_T y^T$, and
\[
(\psi_-^A, y^T \inter y^S, \psi_+^C) \in \pswint_{\tau \inter \sigma}^+(\x_A, \x_B,
\x_C)
\]
with $y^S_B = y^T_B = y_B$, such that the following diagrams commute:
\[
\xymatrix@R=10pt{
&x^S_A  \ar[dl]_{\theta^A_-}
        \ar[dd]^{\omega^S_A}&
x^S_B   \ar[r]^{\theta^B_+}
        \ar[dd]_{\omega^S_B}&
\rep{\x}_B&
x^T_B   \ar[dd]^{\nu^T_B}
        \ar@{<-}[l]_{(\Omega^B_-)^{-1}}&
x^T_C   \ar[dr]^{\Omega^C_+}
        \ar[dd]_{\nu^T_C}\\
\rep{\x}_A&&&&&&\rep{\x}_C\\
&y^S_A  \ar[ul]^{\psi^A_-}&
y^S_B   \ar@{=}[r]&
y_B	\ar@{=}[r]&
y^T_B &
y^T_C   \ar[ur]_{\psi^C_+}
}
\]

We simply set $\Theta_B : \rep{\x}_B \to y_B$ as either path around the center
diagram.
\end{proof}

Thanks to the previous section we may reverse this operation, as shown below.

\begin{lemma}\label{lem:qbip2}
For any symmetry $\Theta_B : \rep{\x}_B \sym_B y_B$ and any witness
\[
(\psi_-^A, y^T \inter y^S, \psi_+^C) \in \pswint_{\tau \inter \sigma}^+(\x_A, \x_B,
\x_C) 
\]
with $y^S_B = y^T_B = y_B$, there are unique symmetries $\omega^S : x^S \sym_S y^S$,
$\nu^T : x^T \sym_T y^T$ and 
\[
(\theta_-^A, x^S, \theta_+^B) \in \pswit_\sigma^+(\x_A, \x_B)\,,
\qquad
(\Omega^B_-, x^T, \Omega^C_+) \in \pswit_\tau^+(\x_B, \x_C)\,,
\]
a pair of causally compatible witnesses, such that the following diagrams commute:
\[
\xymatrix@R=10pt{
&x^S_A  \ar[dl]_{\theta^A_-}
        \ar[dd]^{\omega^S_A}&
x^S_B   \ar[r]^{\theta^B_+}
        \ar[dd]_{\omega^S_B}&
\rep{\x}_B      
        \ar[dd]^{\Theta_B}&
x^T_B   \ar[dd]^{\nu^T_B}
        \ar[l]_{\Omega^B_-}&
x^T_C   \ar[dr]^{\Omega^C_+}
        \ar[dd]_{\nu^T_C}\\
\rep{\x}_A&&&&&&\rep{\x}_C\\
&y^S_A  \ar[ul]^{\psi^A_-}&
y^S_B   \ar@{=}[r]&
y_B&
y^T_B   \ar@{=}[l]&
y^T_C   \ar[ur]_{\psi^C_+}
}
\]
\end{lemma}
\begin{proof}
The first step is to factor $\Theta_B^{-1}$ in two ways, as in the diagram
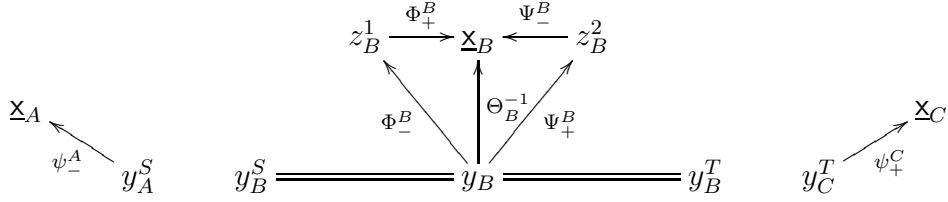
\[
\xymatrix@R=10pt{
&&&z_B^1
	\ar[r]^{\Phi_+^B}&
\rep{\x}_B&
z_B^2	\ar[l]_{\Psi^B_-}\\
\rep{\x}_A&&&&&&&&\rep{\x}_C\\
&y^S_A	\ar[ul]^{\psi_-^A}&
y^S_B	\ar@{=}[rr]&&
y_B	\ar[uul]^{\Phi_-^B}
	\ar[uu]_{\Theta_B^{-1}}
	\ar[uur]_{\Psi_+^B}
	\ar@{=}[rr]&&
y^T_B&
y^T_C	\ar[ur]_{\psi_+^C}
}
\]
following Lemma 3.19 of \cite{cg2}. By Lemma \ref{lem:b4} we can make $\Phi^B_-$ act
on $\sigma$. This yields
\[
\lambda_-^A : x^S_A \sym_A^- y^S_A\,,
\qquad
\omega^S : x^S \sym_S y^S\,,
\qquad
\Delta^B_+ : x^S_B \sym_B^+ z^1_B\,,
\]
unique such that the following diagram commutes:
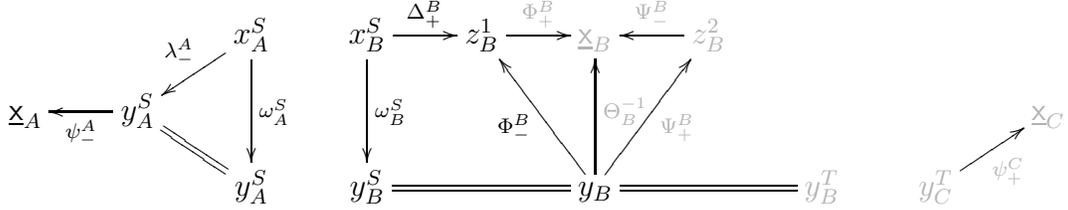
\[
\xymatrix@R=10pt{
&&x^S_A	\ar[dl]_{\lambda^A_-}
	\ar[dd]^{\omega^S_A}&
x^S_B	\ar[dd]^{\omega^S_B}
	\ar[r]^{\Delta^B_+}&
z^1_B	\ar@[grey][r]^{\grey{\Phi_+^B}}&
\grey{\rep{\x}_B}&
\grey{z^2_B}
	\ar@[grey][l]_{\grey{\Psi_-^B}}\\
\rep{\x}_A&
y^S_A	\ar[l]^{\psi^A_-}
	\ar@{=}[dr]&&&&&&&&\grey{\rep{\x}_C}\\
&&y^S_A&y^S_B
	\ar@{=}[rr]&&
y_B	\ar[uul]^{\Phi_-^B}
	\ar@[grey][uu]_{\grey{\Theta_B^{-1}}}
	\ar@[grey][uur]_{\grey{\Psi_+^B}}
	\ar@[grey]@{=}[rr]&&
\grey{y^T_B}&
\grey{y^T_C}
	\ar@[grey][ur]_{\grey{\psi^C_+}}
}
\]
leaving in grey the irrelevant parts of the full diagram for context. Setting
$\theta^A_- = \psi^A_- \circ \lambda^A_-$ and $\theta^B_+ = \Phi^B_+ \circ
\Delta^B_+$, we have found data making the following diagram commute:
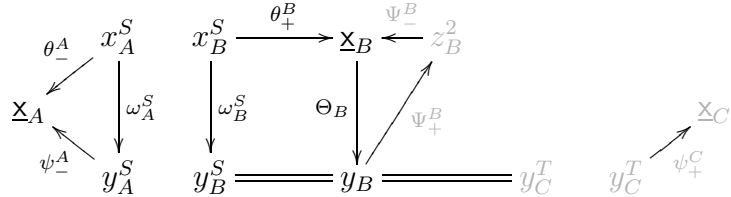
\[
\xymatrix@R=10pt@C=15pt{
&x^S_A	\ar[dl]_{\theta_-^A}
	\ar[dd]^{\omega^S_A}&
x^S_B	\ar[dd]^{\omega^S_B}
	\ar[rr]^{\theta^B_+}&&
\rep{\x}_B
	\ar[dd]_{\Theta_B}&
\grey{z^2_B}
	\ar@[grey][l]_{\grey{\Psi^B_-}}\\
\rep{\x}_A&&&&&&&&\grey{\rep{\x}_C}\\
&y^S_A	\ar[ul]^{\psi_-^A}&
y^S_B	\ar@{=}[rr]&&
y_B	\ar@[grey]@{=}[rr]
	\ar@[grey][uur]_{\grey{\Psi^B_+}}&&
\grey{y^T_C}&
\grey{y^T_C}
	\ar@[grey][ur]_{\grey{\psi^C_+}}
}
\]

We shall now prove uniqueness of this data. Assume that we have other symmetries 
\[
\gamma_-^A : u^S_A \sym_A^- \rep{\x}_A\,,
\qquad
\varpi^S : u^S \sym_S y^S\,,
\qquad
\gamma_+^B : u^S_B \sym_B^+ \rep{\x}_B\,,
\]
making the following diagram commute:
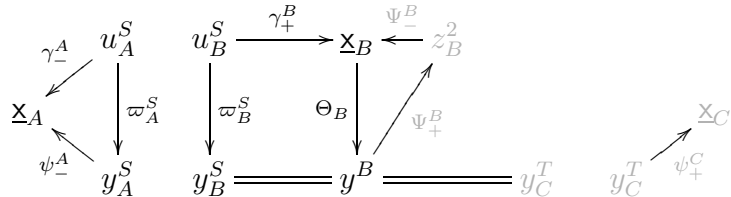
\[
\xymatrix@R=10pt@C=15pt{
&u^S_A  \ar[dl]_{\gamma_-^A}
        \ar[dd]^{\varpi^S_A}&
u^S_B   \ar[dd]^{\varpi^S_B}
        \ar[rr]^{\gamma^B_+}&&
\rep{\x}_B
        \ar[dd]_{\Theta_B}&
\grey{z^2_B}
        \ar@[grey][l]_{\grey{\Psi^B_-}}\\
\rep{\x}_A&&&&&&&&\grey{\rep{\x}_C}\\
&y^S_A  \ar[ul]^{\psi_-^A}&
y^S_B   \ar@{=}[rr]&&
y^B     \ar@[grey]@{=}[rr]
        \ar@[grey][uur]_{\grey{\Psi^B_+}}&&
\grey{y^T_C}&
\grey{y^T_C}
        \ar@[grey][ur]_{\grey{\psi^C_+}}
}
\]

Then, it follows that the following diagram also commutes:
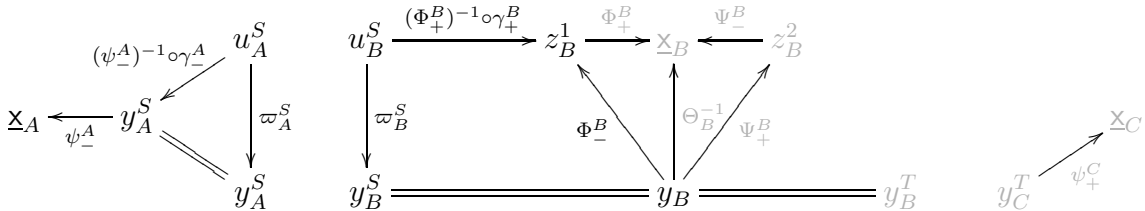
\[
\xymatrix@R=10pt{
&&u^S_A \ar[dl]_{(\psi_-^A)^{-1} \circ \gamma^A_-}
        \ar[dd]^{\varpi^S_A}&
u^S_B   \ar[dd]^{\varpi^S_B}
        \ar[rr]^{(\Phi^B_+)^{-1} \circ \gamma^B_+}&&
z^1_B   \ar@[grey][r]^{\grey{\Phi_+^B}}&
\grey{\rep{\x}_B}&
\grey{z^2_B}
        \ar@[grey][l]_{\grey{\Psi_-^B}}\\
\rep{\x}_A&
y^S_A   \ar[l]^{\psi^A_-}
        \ar@{=}[dr]&&&&&&&&&\grey{\rep{\x}_C}\\
&&y^S_A&y^S_B
        \ar@{=}[rrr]&&&
y_B     \ar[uul]^{\Phi_-^B}
        \ar@[grey][uu]_{\grey{\Theta_B^{-1}}}
        \ar@[grey][uur]_{\grey{\Psi_+^B}}
        \ar@[grey]@{=}[rr]&&
\grey{y^T_B}&
\grey{y^T_C}
        \ar@[grey][ur]_{\grey{\psi^C_+}}
}
\]

By uniqueness for Lemma \ref{lem:b4}, it follows that $u^S = x^S$, $\omega^S =
\varpi^S$,  $\lambda^A_- = (\psi^A_-)^{-1} \circ \gamma_-^A$ so $\gamma_-^A =
\theta_-^A$, and $(\Phi^B_+)^{-1} \circ \gamma^B_+ = \Delta^B_+$ so $\gamma^B_+ =
\theta^B_+$. Altogether, we have proved that there are
\[
\theta_-^A : x^S_A \sym_S \rep{\x}_A\,,
\qquad
\omega^S : x^S \sym_S y^S\,,
\qquad
\theta_+^B : x^S_B \sym_B^+ \rep{\x}_B\,,
\]
unique making the following diagram commutes:
\[
\xymatrix@R=10pt@C=15pt{
&x^S_A  \ar[dl]_{\theta_-^A}
        \ar[dd]^{\omega^S_A}&
x^S_B   \ar[dd]^{\omega^S_B}
        \ar[rr]^{\theta^B_+}&&
\rep{\x}_B
        \ar[dd]_{\Theta_B}&
\grey{z^2_B}
        \ar@[grey][l]_{\grey{\Psi^B_-}}\\
\rep{\x}_A&&&&&&&&\grey{\rep{\x}_C}\\
&y^S_A  \ar[ul]^{\psi_-^A}&
y^S_B   \ar@{=}[rr]&&
y^B     \ar@[grey]@{=}[rr]
        \ar@[grey][uur]_{\grey{\Psi^B_+}}&&
\grey{y^T_C}&
\grey{y^T_C}
        \ar@[grey][ur]_{\grey{\psi^C_+}}
}
\]

The lemma follows by performing the exact same reasoning on the right hand side.
\end{proof}

\subsection{Witnesses of interaction}

With Lemmas \ref{lem:qbip1} and \ref{lem:qbip2} we have done the hardest part of the
job; but to collect the fruits of that work we need to introduce some additional
notation.
We write $\ssym{\x_B}$ for the set of endosymmetries on
$\rep{\x}_B$. Let us fix a choice, for every $x \in \x_B$, of some
\[
\kappa_x : x \sym_B \rep{\x}_B\,.
\]

Transporting through $(\kappa_x)_{x \in \x_B}$ gives a bijection, for any
two $x, y \in \x_B$, between the set of symmetries $x \sym_B y$ and the set
$\ssym{\x_B}$. If $\theta \in \ssym{\x_B}$ and $x, y \in \x_B$, let us write
\[
\theta[x, y] : x \sym_B y
\]
the transported symmetry obtained as $\kappa_y^{-1} \circ \theta \circ \kappa_x$. 

\begin{corollary}
There is a bijection
\[
\Upsilon \quad:\quad \pswit_\sigma^+(\x_A, \x_B) \bullet \pswit_\tau^+(\x_B, \x_C)
\quad\to \quad
\pswint_{\tau \inter \sigma}^+(\x_A, \x_B,\x_C) \times \ssym{\x_B}
\]
such that for every pair of causally compatible witnesses
\[
(\w_1, \w_2) = 
((\theta_-^A, x^S, \theta_+^B),~
(\Omega^B_-, x^T, \Omega^C_+))
\in
\pswit_\sigma^+(\x_A, \x_B) \bullet \pswit_\tau^+(\x_B, \x_C)\,,
\]
writing $((\psi^A_-, y^T\inter y^S, \psi^C_+), \varphi) = \Upsilon(\w_1, \w_2)$, 
$y_B = y^S_B = y^T_B$, $\Theta_B = \varphi[\rep{\x}_B, y_B]$, there are 
$\omega^S : x^S \sym_S y^T$ and $\nu^T : x^T \sym_T y^T$ such that the
following diagrams commute:
\[
\xymatrix@R=10pt{
&x^S_A  \ar[dl]_{\theta^A_-}
        \ar[dd]^{\omega^S_A}&
x^S_B   \ar[r]^{\theta^B_+}
        \ar[dd]_{\omega^S_B}&
\rep{\x}_B      
        \ar[dd]^{\Theta_B}&
x^T_B   \ar[dd]^{\nu^T_B}
        \ar[l]_{\Omega^B_-}&
x^T_C   \ar[dr]^{\Omega^C_+}
        \ar[dd]_{\nu^T_C}\\
\rep{\x}_A&&&&&&\rep{\x}_C\\
&y^S_A  \ar[ul]^{\psi^A_-}&
y^S_B   \ar@{=}[r]&
y_B&
y^T_B   \ar@{=}[l]&
y^T_C   \ar[ur]_{\psi^C_+}
}
\]
\end{corollary}
\begin{proof}
Straightforward from Lemmas \ref{lem:qbip1} and \ref{lem:qbip2}. Note that $\omega^S$
and $\nu^T$ are unique; the requirements of the diagrams constrain them entirely due
to local injectivity of $\sigma, \tau$.
\end{proof}

The commutation of this diagram is required for situations where one would exploit
this in the presence of valuations on configurations that are \emph{typed} and
transported coherently through symmetry, such as for quantum valuations \cite{popl20}.
However, if one is merely interested in \emph{counting} the witnesses, then the take
home message is:

\begin{corollary}\label{cor:cor1}
For any $\x_A \in \sconf{A}, \x_B \in \sconf{B}$ and $\x_C \in \sconf{C}$, we have
\begin{eqnarray}
\card{\pswit_\sigma^+(\x_A, \x_B) \bullet \pswit_\tau^+(\x_B, \x_C)}
&=&
\card{\pswint_{\tau \inter \sigma}^+(\x_A, \x_B,\x_C)} \times
\card{\ssym{\x_B}}\,.\label{eq3}
\end{eqnarray}
\end{corollary}

If we know that the strategies to be composed do not deadlock, then this can be
simplified further.

\begin{corollary}\label{cor:main}
Assume $\sigma : A \stackrel{S}{\to} B$ and $\tau : B \stackrel{T}{\to} C$ do not
deadlock. Then,
\[
\card{\pswit_\sigma^+(\x_A, \x_B)} \times \card{\pswit_\tau^+(\x_B, \x_C)}
=
\card{\pswint_{\tau \inter \sigma}^+(\x_A, \x_B,\x_C)} \times \card{\ssym{\x_B}}\,,
\]
\end{corollary}
\begin{proof}
By hypothesis, causal compatibility is always satisfied. Therefore,
\[
\pswit_\sigma^+(\x_A, \x_B) \bullet \pswit_\tau^+(\x_B, \x_C) 
= 
\pswit_\sigma^+(\x_A, \x_B) \times \pswit_\tau^+(\x_B, \x_C)
\]
and the result follows from Corollary \ref{cor:cor1}.
\end{proof}

This takes us close to Equation \ref{eq1}. One may wonder what is left to conclude; a
hint is the fact that for now, in this section, we have \emph{not used canonicity of
representatives}. 

\subsection{Witnesses and canonicity}

The moral of Equation \ref{eq3} seems clear: on the left hand side
witnesses have the liberty to pick any positive symmetry on respectively $B$ and
$B^\perp$ to interact, whereas on the right hand side they must match on the nose.
Adding $\card{\ssym{\x_B}}$ on the right balances this out.

Let us look deeper into this. From now on, we will rely heavily on canonicity of
representatives. A first consequence of that is the following:

\begin{lemma}\label{lem:canonical}
If $B$ is representable, then for all $\x_B \in \sconf{B}$, we have
\[
\card{\ssym{\x_B}} = \card{\nsym{\x_B}} \times \card{\psym{\x_B}}\,.
\]
\end{lemma}
\begin{proof}
Obvious consequence of the definition of canonicity.
\end{proof}

Indeed this is almost the definition of canonicity, which states that every
endosymmetry on $\rep{\x}_B$ factors uniquely as the composition of a positive and a
negative endosymmetries of $\rep{\x}_B$. Almost as obvious is the following fact:

\begin{lemma}\label{lem:elim_sym}
For any $\sigma : A \stackrel{S}{\to} B, \tau : B \stackrel{T}{\to} C$, 
$\x_A \in \sconf{A}$, $\x_B \in \sconf{B}$, and $\x_C \in \sconf{C}$, 
\begin{eqnarray*}
\card{\pswit_\sigma^+(\x_A, \x_B)} &=& \card{\nsym{\x_A}} \times
\card{\wit_\sigma^+(\x_A, \x_B)} \times \card{\psym{\x_B}}\\
\card{\pswint^+_{\tau \inter \sigma}(\x_A, \x_B, \x_C)} &=& 
\card{\nsym{\x_A}} \times \card{\wint^+_{\tau \inter \sigma}(\x_A, \x_B, \x_C)}
\times \card{\psym{\x_C}}
\end{eqnarray*}
\end{lemma}
\begin{proof}
We only detail the first equality, the reasoning for the other is identical.
Let us choose, for every $x \in \x_B$ such that $x \sym_B^+ \rep{\x}_B$, some
positive symmetry $\kappa^B_x : \rep{\x}_B \sym_B^+ x$. Likewise we choose, for
each $y \in \x_A$ such that $y \sym_A^- \rep{\x}_A$, some $\kappa^A_y : \rep{\x}_A
\sym_A^- y$.

Now, we form the function:
\[
\begin{array}{rcrcl}
G &:& \pswit_\sigma^+(\x_A, \x_B) &\to & \nsym{\x_A} \times \wit_\sigma^+(\x_A, \x_B)
\times \psym{\x_B}\\
&&(\theta_-^A, x^S, \theta_+^B) &\mapsto& 
(\theta_-^A \circ \kappa^A_{x^S_A}, ~x_S,~\theta_+^B \circ \kappa^B_{x^S_B})
\end{array}
\]
which is clearly a bijection as positive and negative symmetries are invertible.
\end{proof}

\subsection{Wrapping up}

Finally, we are now in position to prove:

\begin{theorem}
Consider $\sigma : A \stackrel{S}{\to} B$ and $\tau : B \stackrel{T}{\to} C$ that do
not deadlock, and assume that $B$ is representable. 
Then, for all $\x_A \in \sconf{A}, \x_C \in \sconf{C}$, we have
\[
(\coll (\tau \odot \sigma))_{\x_A, \x_C} = 
\sum_{\x_B \in \sconf{B}} (\coll \sigma)_{\x_A, \x_B} \cdot (\coll \tau)_{\x_B,
\x_C}\,.
\]
\end{theorem}
\begin{proof}
We calculate
\begin{eqnarray}
(\coll(\tau \odot \sigma))_{\x_A, \x_C} &=&
\card{\wit_{\tau \odot \sigma}^+(\x_A, \x_C)}
\label{eq4}\\
&=& \card{\wint^+_{\tau \inter \sigma}(\x_A, \x_C)}
\label{eq5}\\
&=& \sum_{\x_B \in \sconf{B}} \card{\wint^+_{\tau \inter \sigma}(\x_A, \x_B, \x_C)}
\label{eq6}\\
&=& \sum_{\x_B \in \sconf{B}} \frac{\card{\pswint^+_{\tau
\inter \sigma}(\x_A, \x_B, \x_C)}}{\card{\nsym{A}}\cdot \card{\psym{C}}}
\label{eq7}\\
&=& \sum_{\x_B \in \sconf{B}} \frac{\card{\pswit^+_{\sigma}(\x_A, \x_B)} \cdot
\card{\pswit^+_\tau(\x_B, \x_C)}}{\card{\nsym{A}}\cdot \card{\ssym{B}} \cdot
\card{\psym{C}}}
\label{eq8}\\
&=& \sum_{\x_B \in \sconf{B}} \frac{\card{\psym{\x_B}}\cdot
\card{\nsym{\x_B}}}{\card{\ssym{\x_B}}} \cdot \card{\wit^+_{\sigma}(\x_A, \x_B)}
\cdot \card{\wit^+_\tau(\x_B, \x_C)}
\label{eq9}\\
&=& \sum_{\x_B \in \sconf{B}}
\card{\wit^+_{\sigma}(\x_A, \x_B)} \cdot \card{\wit^+_\tau(\x_B, \x_C)}
\label{eq10}\\
&=& \sum_{\x_B \in \sconf{B}} (\coll \sigma)_{\x_A, \x_B} \times (\coll \tau)_{\x_B,
\x_C}\label{eq11}
\end{eqnarray}
where \eqref{eq4} is by definition, \eqref{eq5} is by Lemma \ref{lem:main_pcov},
\eqref{eq6} is by Lemma \ref{lem:part_int}, \eqref{eq7} is by Lemma
\ref{lem:elim_sym}, \eqref{eq8} is by Corollary \ref{cor:main}, \eqref{eq9} is by
Lemma \ref{lem:elim_sym} again, \eqref{eq10} is by Lemma \ref{lem:canonical}
exploiting that $B$ is representable; and \eqref{eq11} is by definition.
\end{proof}

This concludes the proof of \eqref{eq1}. 

\section{Epilogue: back to symmetry classes}
\label{sec:epilogue}

To conclude, we show that our original notion of witness based on symmetry
classes rather than canonical representatives was, in fact, wrong. We give two
counter-examples: the first is geared towards simplicity, while the second aims to
bring the counter-example as close as possible to usual models of programming
languages. The two examples are, however, powered by the same phenomenon.

\begin{example}
Consider the games $A^\perp, C = \done$ formed of only one positive move. The game $B =
\ominus_1 \ominus_2 \oplus$ has three moves, with $\ominus_1$ and $\ominus_2$
symmetric. We consider two strategies:
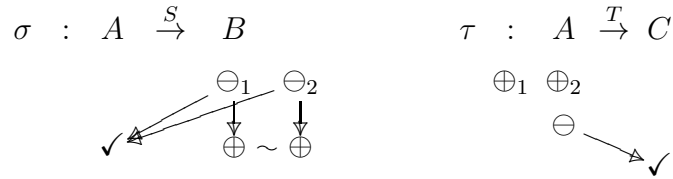
\[
\xymatrix@C=5pt@R=2pt{
\sigma &:& A & \stackrel{S}{\to} & B\\
&&&&\ominus_1
	\ar@{-|>}[ddll]
	\ar@{-|>}[dd]&
\ominus_2
	\ar@{-|>}[ddlll]
	\ar@{-|>}[dd]\\\\
&&\done&&\oplus
	\ar@{~}[r]&\oplus
}
\qquad
\qquad
\xymatrix@C=0pt@R=2pt{
\tau &:& A & \stackrel{T}{\to} & C\\
&\oplus_1 & \oplus_2\\
&&\ominus\ar@{-|>}[drr]\\
&&&&\done
}
\]

These are indeed valid strategies in the sense of \cite{cg2}. Their composition is:
\[
\xymatrix@C=-2pt@R=0pt{
\tau \odot \sigma &:& A &\stackrel{T\odot S}{\to}&& C\\
&&\done &&\done
	\ar@{~}[rr]&&
\done
}
\]

In particular, the non-deterministic choice on the right hand side originates from
the choice by $\sigma$: to which $\ominus_i$ should it react? In particular,
\[
\card{\wit_{\tau\odot \sigma}(\{\done\}, \{\done\})} = 2\,,
\]
as the two occurrences of $\done$ on the right are not symmetric (this boils down to
the fact that $\oplus_1$ and $\oplus_2$ cannot be symmetric in $\tau$, by thinness). 
On the other hand, the only symmetry class of $B$ on which these two may interact is
$\{\ominus_1, \ominus_2, \oplus\}$. And we have:
\[
\card{\wit_{\sigma}(\{\done\}, \{\ominus_1, \ominus_2, \oplus\})} = 1
\qquad
\qquad
\card{\wit_{\tau}(\{\oplus_1, \oplus_2, \ominus\}, \{\done\})} = 1
\]

In particular, the two configurations of $\sigma$ responsible for the
non-deterministic choice
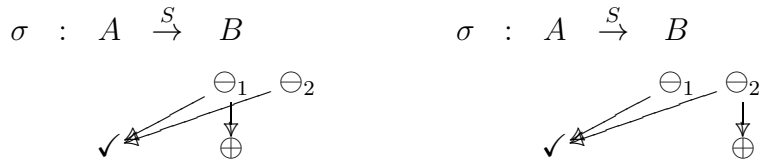
\[
\xymatrix@C=5pt@R=2pt{
\sigma &:& A & \stackrel{S}{\to} & B\\
&&&&\ominus_1
        \ar@{-|>}[ddll]
        \ar@{-|>}[dd]&
\ominus_2
        \ar@{-|>}[ddlll]\\\\
&&\done&&\oplus
}
\qquad
\qquad
\xymatrix@C=5pt@R=2pt{
\sigma &:& A & \stackrel{S}{\to} & B\\
&&&&\ominus_1
        \ar@{-|>}[ddll]&
\ominus_2
        \ar@{-|>}[ddlll]
        \ar@{-|>}[dd]\\\\
&&\done&&&\oplus
}
\]
are symmetric, so they form only one symmetry class and are counted only once in
$\wit_\sigma(\{\done\}, \{\ominus_1, \ominus_2, \oplus\})$ -- whereas they are two
distinct elements of $\wit_\sigma^+(\{\done\}, \{\ominus_1, \ominus_2, \oplus\})$.

This also shows that it is \emph{not} the case that configurations in
$\wit_\sigma^+(x)$ are canonical representatives of symmetry classes -- they are
better than that, as they get it right where symmetry classes get it wrong.
\end{example}

We now show essentially the same example in a more ``programming language'' style.

\begin{example}\label{ex:counter_ex2}
Consider a basic game $o$, with a unique move $\qu^-$.
Consider strategies:
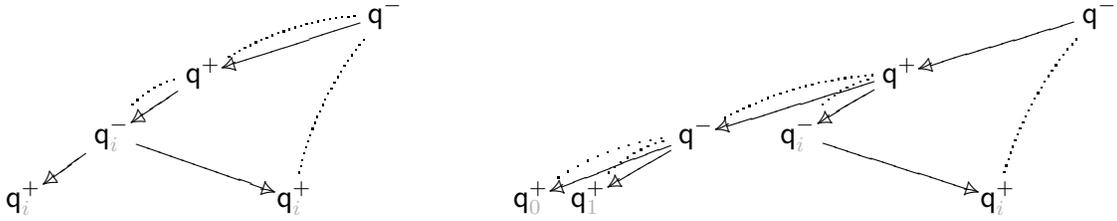
\[
\xymatrix@R=5pt@C=-2pt{
\sigma &:& \oc o &\stackrel{S}{\to}& (\oc o & \lin& o) &\lin & \oc o & \lin & o\\
&&&&&&&&&&\qu^-
	\ar@{-|>}[dllll]\\
&&&&&&\qu^+
	\ar@{-|>}[dll]
	\ar@{.}@/^/[urrrr]\\
&&&&\qu^-_{\grey{i}}
	\ar@{.}@/^/[urr]
	\ar@{-|>}[dll]
	\ar@{-|>}[drrrr]\\
&&\qu^+_{\grey{i}}&&&&&&
\qu^+_{\grey{i}}
	\ar@{.}@/^/[uuurr]
}
\qquad
\xymatrix@R=5pt@C=0pt{
\tau &:& ((\oc o& \lin & o) &\lin &\oc o& \lin & o) & \stackrel{T}{\to}& \oc o & \lin
& o\\
&&&&&&&&&&&&\qu^-
	\ar@{-|>}[dllll]\\
&&&&&&&&\qu^+
	\ar@{-|>}[dllll]
	\ar@{-|>}[dll]\\
&&&&\qu^-
	\ar@{.}@/^/[urrrr]
	\ar@{-|>}[dlll]
	\ar@{-|>}[dll]&&
\qu^-_{\grey{i}}
	\ar@{.}@/^/[urr]
	\ar@{-|>}[drrrr]\\
&\qu^+_{\grey{0}}
	\ar@{.}@/^/[urrr]&
\qu^+_{\grey{1}}
	\ar@{.}@/^/[urr]&&&&&&&&
\qu^+_{\grey{i}}
	\ar@{.}@/^/[uuurr]
}
\]

Note that the moves on the left are only there to ensure the $+$-covered hypothesis.
Their composition is:
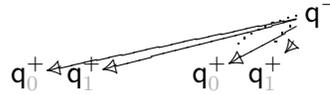
\[
\xymatrix@R=5pt@C=2pt{
\tau \odot \sigma &:& \oc o & \stackrel{T\odot S}{\to} & \oc o & \lin & o\\
&&&&&&\qu^-
	\ar@{-|>}[dlllll]
	\ar@{-|>}[dllll]
	\ar@{-|>}[dll]
	\ar@{-|>}[dl]\\
&\qu^+_{\grey{0}}&
\qu^+_{\grey{1}}&&
\qu^+_{\grey{0}}
	\ar@{.}@/^/[urr]&
\qu^+_{\grey{1}}
	\ar@{.}@/^/[ur]
}
\]

Now, as for the previous example, we observe:
\[
\card{\wit_{\tau\odot \sigma}([\qu, \qu], [\qu] \lin \qu)} = 2
\]
using an intersection type like notation for symmetry classes on the game, which
hopefully is clear. The reader may check that there is a unique
symmetry class on $(\oc o \lin o ) \lin \oc o \lin o$ on which the strategies may
match to produce this via $+$-covered configurations, namely
\[
([\qu, \qu] \lin \qu) \lin [\qu] \lin \qu
\]
in the same intersection-type like notation. And we have
\[
\card{\wit_\sigma([\qu, \qu], ([\qu, \qu] \lin \qu) \lin [\qu] \lin \qu)} = 1
\qquad
\card{\wit_\tau(([\qu, \qu] \lin \qu) \lin [\qu] \lin \qu), ([\qu], \qu))} = 1
\]
for the same reason as in the previous example.
\end{example}

These strategies are not quite terms but they are very well-behaved, in particular
visible and parallel innocent. This counter-example does not quite contradict the
claims of \cite{lics18} because there strategies are more constrained (in particular
they are well-bracketed) and the positions of interest (matching the
points of the web) are complete. It is plausible that this makes this pathology
disappear -- in particular Example \ref{ex:counter_ex2} exploits non-well bracketed
behaviour, but this is pure speculation. In any case concrete witnesses as developped
here are definitely better behaved, and are recommended in all situations.

\paragraph{Acknowledgments.} This work is supported by ANR project DyVerSe
(ANR-19-CE48-0010-01) and Labex MiLyon (ANR-10-LABX-0070) of Universit\'e de Lyon,
within the program ``Investissements d'Avenir'' (ANR-11-IDEX-0007), operated by the
French National Research Agency (ANR). 

\bibliographystyle{plain}
\bibliography{main}

\end{document}

%% file: macros.tex
\newcommand{\inter}{\circledast}

\newcommand{\lin}{\multimap}

\newcommand{\id}{\mathrm{id}}

\newcommand{\conf}[1]{\mathscr{C}(#1)}
\newcommand{\pconf}[1]{\mathscr{C}^+(#1)}
\newcommand{\imc}{\rightarrowtriangle}
\newcommand{\mconflict}{\!\!\xymatrix@C=10pt{~\!\!\ar@{~}[r]&~\!\!}\!\!}
\def\profto{\mathrel{\mkern3mu  \vcenter{\hbox{$\scriptscriptstyle+$}}%
                    \mkern-12mu{\to}}}

\newcommand{\iso}{\cong}

\newcommand{\sym}{\mathrel{\cong}}

\newcommand{\w}{\mathsf{w}}

\newdir{|>}{!/4.5pt/@{|}*:(1,-.2)@{>}*:(1,+.2)@_{>}}

\newcommand{\wit}{\mathsf{wit}}
\newcommand{\x}{\mathsf{x}}
\newcommand{\y}{\mathsf{y}}
\newcommand{\z}{\mathsf{z}}
\newcommand{\rep}[1]{\underline{#1}}

\newdir{pb}{:(1,-1)@^{|-}}
\def\pb#1{\save[]+<16 pt,0 pt>:a(#1)\ar@{pb{}}[]\restore}
\newcommand{\canon}{\mathsf{canon}}

\newcommand{\coll}[1]{\mathop{\smallint} #1}

\newcommand{\Rel}{\mathsf{Rel}}
\newcommand{\sconf}[1]{\mathscr{C}_{\!\sym}(#1)}
\newcommand{\spconf}[1]{\mathscr{C}^+_{\!\sym}(#1)}
\newcommand{\mset}[1]{\mathscr{M}_f(#1)}

\newcommand{\lbl}{\mathrm{lbl}}
\newcommand{\ind}{\mathrm{ind}}
\newcommand{\nsym}[1]{\mathscr{S}_{\!-}({#1})}
\newcommand{\ssym}[1]{\mathscr{S}({#1})}
\newcommand{\psym}[1]{\mathscr{S}_{\!+}({#1})}
\newcommand{\swit}{\text{$\sim$-$\wit$}}
\newcommand{\pswit}{\text{$\sim^{\!\!+}\!\!$-$\wit$}}
\newcommand{\wint}{\mathsf{int}}
\newcommand{\pswint}{\text{$\sim^{\!\!+}\!\!$-$\wint$}}
\def\acts{\curvearrowright}
\newcommand{\done}{\checkmark}
\newcommand{\swap}{\mathsf{sw}}
\newcommand{\qu}{\mathsf{q}}